\documentclass[twoside,11pt]{article}
\usepackage{amsmath}
\usepackage{mathrsfs}
\usepackage{yhmath}
\usepackage{amsxtra}
\usepackage{amscd}
\usepackage{amsthm}
\usepackage{amsfonts}
\usepackage{amssymb}
\usepackage{anysize}
\usepackage{bbding}
\usepackage{bbm}
\usepackage[bookmarks=true]{hyperref}
\usepackage{eucal}
\usepackage{enumerate}
\usepackage{fancyhdr}
\usepackage{mathtools,graphics,latexsym}
\usepackage{pifont}
\usepackage{pmgraph}
\usepackage{stmaryrd}
\usepackage[usenames,dvipsnames]{color}
\usepackage{tikz-cd}

\marginsize{3cm}{3cm}{2.5cm}{2.5cm}

\newtheorem{thm}{Theorem}[section]
\newtheorem{cor}[thm]{Corollary}
\newtheorem{lem}[thm]{Lemma}
\newtheorem{prop}[thm]{Proposition}

\theoremstyle{definition}
\newtheorem{defn}[thm]{Definition}

\theoremstyle{remark}
\newtheorem{rem}[thm]{Remark}

\numberwithin{equation}{section}

\begin{document}

\newcommand{\thmref}[1]{Theorem~\ref{#1}}
\newcommand{\secref}[1]{Section~\ref{#1}}
\newcommand{\lemref}[1]{Lemma~\ref{#1}}
\newcommand{\propref}[1]{Proposition~\ref{#1}}
\newcommand{\corref}[1]{Corollary~\ref{#1}}
\newcommand{\remref}[1]{Remark~\ref{#1}}
\newcommand{\eqnref}[1]{(\ref{#1})}
\newcommand{\exref}[1]{Example~\ref{#1}}

\makeatletter
\newsavebox\myboxA
\newsavebox\myboxB
\newlength\mylenA

\newcommand*\xoverline[2][0.75]{
    \sbox{\myboxA}{$\m@th#2$}
    \setbox\myboxB\null
    \ht\myboxB=\ht\myboxA
    \dp\myboxB=\dp\myboxA
    \wd\myboxB=#1\wd\myboxA
    \sbox\myboxB{$\m@th\overline{\copy\myboxB}$}
    \setlength\mylenA{\the\wd\myboxA}
    \addtolength\mylenA{-\the\wd\myboxB}
    \ifdim\wd\myboxB<\wd\myboxA
       \rlap{\hskip 0.5\mylenA\usebox\myboxB}{\usebox\myboxA}
    \else
        \hskip -0.5\mylenA\rlap{\usebox\myboxA}{\hskip 0.5\mylenA\usebox\myboxB}
    \fi}
\makeatother

\makeatletter
\def\mathcenterto#1#2{\mathclap{\phantom{#1}\mathclap{#2}}\phantom{#1}}
\let\old@wt\widetilde
\def\widetildeto#1#2{\mathcenterto{#2}{\old@wt{\mathcenterto{#1}{#2}}}}
\makeatother

\def\wt{\widetildeto{o}}

\DeclarePairedDelimiterX\setc[2]{\{}{\}}{\,#1 \;\delimsize\vert\; #2\,}

\newcommand{\nc}{\newcommand}
\nc{\Z}{{\mathbb Z}}
\nc{\C}{{\mathbb C}}
\nc{\N}{{\mathbb N}}
\nc{\la}{\lambda}
\nc{\ep}{\epsilon}
\nc{\La}{\Lambda}
\nc{\V}{\mf V} \nc{\bi}{\bibitem}
\nc{\mc}{\mathcal}
\nc{\mf}{\mathfrak}
\nc{\hf}{\frac{1}{2}}
\nc{\gl}{{\mf{gl}}}
\nc{\osp}{\mf{osp}}
\nc{\spo}{\mf{spo}}
\nc{\ov}{\overline}
\nc{\xov}{\xoverline}
\nc{\un}{\underline}
\nc{\w}{\widetilde}
\nc{\J}{\mathbb{J}}
\nc{\xx}{{\mf x}}
\nc{\OO}{\mathcal{O}}
\nc{\Hom}{\textrm{Hom}}

\nc{\wla}{\w{\la}}
\nc{\ovla}{\ov{\la}}
\nc{\wmu}{\w{\mu}}
\nc{\ovmu}{\ov{\mu}}

\nc{\cP}{{\mathcal{P}}}
\nc{\cH}{{\mathcal{H}}}
\nc{\ovcH}{\ov{{\mathcal{H}}}}
\nc{\wcH}{\w{\mathcal{H}}}

\nc{\cG}{{\mathcal{G}}}
\nc{\ovcG}{\ov{\mathcal{G}}}
\nc{\wcG}{\w{\mathcal{G}}}

\nc{\G}{{\mathfrak g}}
\nc{\SG}{\ov{\mathfrak g}}
\nc{\DG}{\tilde{\mathfrak g}}

\nc{\fh}{\mathfrak h}
\nc{\ovfh}{\ov{{\mathfrak h}}}
\nc{\wfh}{\tilde{{\mathfrak h}}}

\setlength\headheight{14pt}

\pagestyle{fancy}
\fancyhf{}
\fancyhead[RO]{\thepage}
\fancyhead[LE]{\thepage}
\fancyhead[CO]{\small{Gaudin Hamiltonians on unitarizable modules}}
\fancyhead[CE]{\small{Cheong and Lam}}

\title{\Large \bf Gaudin Hamiltonians on unitarizable modules over classical Lie (super)algebras}
\author{Wan Keng Cheong and Ngau Lam}
\date{}

\maketitle

\begin{abstract}

Let $M$ be a tensor product of unitarizable irreducible highest weight modules over the Lie (super)algebra $\cG$, where $\cG$ is $\gl(m|n)$, $\osp(2m|2n)$ or $\spo(2m|2n)$. We show, using super duality, that the singular eigenvectors of the (super) Gaudin Hamiltonians for $\cG$ on $M$ can be obtained from the singular eigenvectors of the Gaudin Hamiltonians for the corresponding Lie algebras on some tensor products of finite-dimensional irreducible modules. As a consequence, the (super) Gaudin Hamiltonians for $\cG$ are diagonalizable on the space spanned by singular vectors of $M$ and hence on $M$. In particular, we establish the diagonalization of the Gaudin Hamiltonians, associated to any of the orthogonal Lie algebra $\mf{so}(2n)$ and the symplectic Lie algebra $\mf{sp}(2n)$, on the tensor product of infinite-dimensional unitarizable irreducible highest weight modules.

\end{abstract}

 \setcounter{tocdepth}{1}

 \section{Introduction}

The Gaudin model was introduced by Gaudin \cite{G76, G83} to describe a completely integrable quantum spin chain associated to any finite-dimensional simple Lie algebra $\cG$. Let $(\cdot, \cdot)$ be a non-degenerate invariant bilinear form on $\cG$. Let $\{I_a \, | \, a=1, \ldots, d\}$ be a basis for $\cG$ and $\{I^a\, |\, a=1, \ldots, d\}$ the dual basis with respect to the bilinear form $(\cdot, \cdot)$, where $d$ is the dimension of $\cG$.
The Casimir symmetric tensor $\Omega$ for $\cG$ is defined to be $\Omega=\sum_{a=1}^d I_a\otimes I^a$. Fix $\ell \in \N$ with $\ell \geq 2$, and let $z_1, \ldots, z_\ell$ be distinct complex numbers. For any $i=1, \ldots, \ell$, the (quadratic) Gaudin Hamiltonian $H^i$ of the Gaudin model associated to $\cG$ is defined to be
\begin{equation*}
H^i=\sum_{\substack{j=1 \\ j\not=i}}^\ell \frac{\Omega^{(ij)}}{z_i-z_j},
\end{equation*}
where $\Omega^{(ij)}$ is defined as in \eqref{(ij)}.
The Gaudin Hamiltonians $H^1, \ldots, H^\ell$ act on the tensor product $M_1\otimes \cdots \otimes M_\ell$, where each $M_i$ is a $\cG$-module, and they are mutually commuting operators.

To find common eigenvectors for Gaudin Hamiltonians is one of the main problems of studying the Gaudin model. Bethe ansatz method provides an explicit construction of common eigenvectors from the solutions of the so-called Bethe ansatz equations and proves to be effective for the special linear Lie algebra $\mf{sl}(2)$ (cf. \cite{G76}). The eigenvectors obtained by this method are called Bethe vectors. Babujian and Flume \cite{BF} generalized the Bethe ansatz equations to the case of any simple Lie algebra. In \cite{FFR}, Feigin, Frenkel and Reshetikhin proposed a new approach based on Wakimoto modules over the affine Lie algebra at the critical level. They obtained the Bethe vectors by restricting certain invariant functionals on tensor products of Wakimoto modules and found as well that the Bethe vectors are in close connection with the solutions of the Knizhnik-Zamolodchikov equations \cite{KZ} (see also \cite{EFK}).

In this paper, we are interested in the super version of (quadratic) Gaudin Hamiltonians. For precise definitions of such Gaudin Hamiltonians, see \eqref{Hamiltonian-C} and \eqref{Hamiltonian}. We find an interesting connection between the Gaudin Hamiltonians associated to the (trivial) central extension of any infinite-rank classical Lie (super)algebra of type $\mf{a, c, d}$ and the Gaudin Hamiltonians associated to the (trivial) central extension of the corresponding Lie algebra.

Let $\w\cG$ be the infinite-rank classical Lie superalgebra of type $\mf{a, c, d}$ described in Sections \ref{type-a}--\ref{type-d}, and let $\cG[m]$ and $\ov\cG[m]$ be certain subalgebras of $\w\cG$. Let $\w\G$, $\G[m]$ and $\SG[m]$ be the central extensions of $\w\cG$, $\cG[m]$ and $\ov\cG[m]$, respectively.
Our main tool is super duality (cf. \cite{CL10, CLW11, CLW12}), which asserts that there are equivalences of tensor categories between the parabolic BGG categories $\w\OO$ of $\DG$-modules, $\OO[m]$ of $\G[m]$-modules and $\ov\OO[m]$ of $\SG[m]$-modules (cf. \thmref{SD}). The following diagrams summarize the relationships among these categories.
$$
\begin{tikzcd}[column sep=small]
& \w\OO \arrow[dl, "T_{[m]}"'] \arrow[dr, "T_{[0]}"] & \\
\OO[m]  &    &  \OO[0]
\end{tikzcd}
\qquad
\begin{tikzcd}[column sep=small]
& \w\OO \arrow[dl, "\ov{T}_{[m]}"'] \arrow[dr, "\ov{T}_{[0]}"] & \\
\ov\OO[m]  &    &  \ov\OO[0]
\end{tikzcd}
$$
Here the tensor functors $T_{[m]}$, $\ov{T}_{[m]}$, $T_{[0]}$, and $\ov{T}_{[0]}$, defined in Section \ref{category-O}, are equivalences of tensor categories for $m \in \N$. Note that $\DG$, $\G[m]$ and $\SG[m]$, for $m \in \N$, are Lie superalgebras while $\G[0]$ and $\SG[0]$ are Lie algebras. In other words, super duality gives equivalences of categories between the parabolic BGG categories for Lie superalgebras and the corresponding Lie algebras.

We view the Gaudin Hamiltonians for $\DG$ (resp., $\G[m]$ and $\SG[m]$) as linear operators on the tensor product of modules in the category $\w\OO$ (resp., $\OO[m]$ and $\ov\OO[m]$). Instead of constructing singular eigenvectors for these operators directly, we apply super duality and the work \cite{CaL} to establish the following result. 

\begin{thm}[\thmref{eigenvector-C1}]
There exist one-to-one correspondences relating the sets of singular eigenvectors of the Gaudin Hamiltonians for $\DG$, $\G[m]$ and $\SG[m]$ for $m \in \Z_+$.
\end{thm}

Furthermore, by using the notion of truncation functors (see \eqref{trun-functor}), we relate the singular eigenvectors of the Gaudin Hamiltonians for $\G[m]$ and $\SG[m]$ to the singular eigenvectors of the Gaudin Hamiltonians for their finite-rank counterparts $\G[m]_n$ and $\SG[m]_n$, for $n \in \N$.

The diagonalization problem of Gaudin Hamiltonians plays an important role in the study of Gaudin models. Let $\cG$ be a finite-dimensional simple Lie algebra. Rybnikov \cite{Ry} shows that for generic $z_1, \ldots, z_\ell$, the Gaudin algebra, which is generated by (quadratic) Gaudin Hamiltonians and higher Gaudin Hamiltonians, is diagonalizable on the space spanned by singular vectors in $L_1\otimes \cdots \otimes L_\ell$, where $L_i$ is a finite-dimensional irreducible module over $\cG$ for each $i=1, \ldots, \ell$. We refer the reader to \cite{MV, MTV1, MTV2, FFRyb, LMV} for related results.

In this paper, we would like to extend Rybnikov's result to the case of Lie (super)algebras. We focus our attention on the finite-dimensional classical Lie (super)algebra $\ovcG^\xx[m]_n$ of type $\xx$, where $\xx$ denotes a fixed type among $\mf{a, c,d}$ (see Sections \ref{type-a}, \ref{type-c} and \ref{type-d}). Note that $\ovcG^{\mf{a}}[m]_n \cong \gl(m|n)$, $\ovcG^{\mf{c}}[m]_n \cong \mf{spo}(2m|2n)$ and $\ovcG^{\mf{d}}[m]_n \cong \mf{osp}(2m|2n)$.

\begin{thm} [\thmref{dom-uni-1} and \thmref{dom-uni-2}]
Let $M=\ov{L}_1 \otimes \cdots \otimes \ov{L}_\ell$, where each $\ov{L}_i$ is an irreducible highest weight module over $\ovcG^\xx[m]_n$, and let $M^{\rm sing}$ be the space spanned by singular vectors in $M$.
Suppose that for $\xx=\mf{a}$, each $\ov{L}_i$ is a polynomial module, and for $\xx=\mf{c,d}$, all the highest weights of $\ov{L}_i$'s lie in $\ov{\mc Q}^{\xx, I}(m|n)$ or $\ov{\mc Q}^{\xx, II}(m|n)$ (see Section \ref{unitarizable-cd}). Then for $i=1, \ldots, \ell$ and generic $z_1, \ldots, z_\ell$, the (super) Gaudin Hamiltonian $\xov H^i[m]_n$ associated to $\ovcG^\xx[m]_n$ is diagonalizable on $M^{\rm sing}$, and any of its eigenbasis can be obtained from the singular eigenvectors of the Gaudin Hamiltonians associated to the Lie algebra counterparts $\cG^\xx[0]_k$'s on some tensor products of finite-dimensional irreducible modules.
\end{thm}

\begin{rem} 
The $\ov{L}_i$'s in the above theorem are unitarizable modules over $\ovcG^\xx[m]_n$ with respect to a natural choice of $*$-structures on $\ovcG^\xx[m]_n$ (see Sections \ref{unitarizable-a} and \ref{unitarizable-cd}). The theorem should also give a procedure of constructing common eigenvectors of the Gaudin Hamiltonians for classical Lie (super)algebras from the ones for the corresponding Lie algebras.

\end{rem}

While there is a great deal of work on Gaudin models, not much appears to be known about the super Gaudin Hamiltonians. In \cite{MVY}, Mukhin, Vicedo and Young investigate the Gaudin Hamiltonians associated to the general linear Lie superalgebra $\gl(m|n)$ and show that the operators are diagonalizable on the space spanned by singular vectors in the tensor power of the natural module $\C^{m|n}$ with an eigenbasis formed by Bethe vectors. In \cite{KuM}, Kulish and Manojlovi\'{c} explicitly construct eigenvectors for the Gaudin Hamiltonians associated to the ortho-symplectic Lie superalgebra $\mf{osp}(1|2)$. The reader may also want to consult \cite{HMVY, L, LM} for related results. We hope that this paper can provide new insights into Gaudin models.

This paper is organized as follows. In Section \ref{pre}, we fix notation and review some background materials on classical Lie (super)algebras and their central extensions. We also define the parabolic BGG categories $\w\OO$, $\OO[m]$ and $\ov\OO[m]$ associated to the Lie (super)algebras $\DG$, $\G[m]$ and $\SG[m]$, respectively, and discuss super duality which gives equivalences of these categories. In Section \ref{unitarizable}, we give a brief introduction to $*$-structures and study the unitarizable modules which will be considered in our study of Gaudin Hamiltonians. In Section \ref{GH}, we investigate the Gaudin Hamiltonians associated to the central extensions of finite-rank and infinite-rank Lie (super)algebras, and establish the one-to-one correspondences between the sets of singular eigenvectors (see \thmref{eigenvector-C1} and \propref{trun_eigenvector}). In Section \ref{diag}, we concentrate on the Gaudin Hamiltonians for finite-dimensional classical Lie (super)algebras and prove \thmref{dom-uni-1} and \thmref{dom-uni-2}, which give an affirmative answer to the diagonalization of the operators on unitarizable modules and a possible way of finding common eigenvectors.

\vskip 0.3cm
\noindent{\bf Notations.} Throughout the paper, $\N$ stands for the set of positive integers, $\Z$ for the set of integers, $\Z^*$ for the set of nonzero integers, $\Z_+$ for the set of non-negative integers, $\hf\Z$ for the set of half integers and integers, and $\C$ for the set of complex numbers.  All vector spaces, algebras, tensor products, etc., are over $\C$.

\bigskip

\section{Preliminaries} \label{pre}
 In this section, we first define the finite-rank and infinite-rank Lie (super)algebras $\wcG^\xx$, $\cG^\xx[m]_n$ and $\ovcG^\xx[m]_n$, where $\xx$ denotes one of the three types $\mf{a,c,d}$.  We consider their central extensions $\wt\G^\xx$, $\G^\xx[m]_n$ and ${\SG}^\xx[m]_n$ and the parabolic BGG category $\w\OO^\xx$ (resp., $\OO^\xx[m]_n$ and $\ov\OO^\xx[m]_n$) of modules over $\wt\G^\xx$  (resp., $\G^\xx[m]_n$ and ${\SG}^\xx[m]_n$).
We then recall the truncation functors which relate $\OO^\xx[m]_\infty$ and $\ov\OO^\xx[m]_\infty$ to $\OO^\xx[m]_n$ and $\ov\OO^\xx[m]_n$, respectively, for $n \in \N$.
Finally, we describe the tensor functors $T_{[m]}$ and $\ov T_{[m]}$ and their properties. We refer the readers to \cite[Sections 2 and 3]{CL10} for type $\mf{a}$ and \cite[Sections 2 and 3]{CLW11} for types $\mf{c,d}$
for details (see also \cite[Sections 6.1 and 6.2]{CW} and \cite[Section 2.4]{CaL}).
{\bf  We fix $m\in\Z_+$ and $n \in \N\cup\{\infty\}$ throughout this paper.}

Let $\w{V}$ denote the superspace over $\C$ with ordered basis $\{v_r \,|\, r \in\hf\Z\}$. The parity of $v_r$ is defined as follows: $|v_r|=\bar{0}$ if $r\in\Z$, and $|v_r|=\bar{1}$ if $r\in\hf+\Z$.

Let $\gl(\w{V})$ be the Lie superalgebra consisting of all linear endomorphisms on $\w{V}$ which vanish on all but finitely many $v_r$'s. For $i, j \in \hf\Z$, we let $E_{i, j}$ be the linear endomorphism on $\w{V}$ defined by
$$E_{i, j}(v_r)=\delta_{j,r} v_i  \qquad \textrm{for }r \in \hf\Z,$$
where $\delta$ is the Kronecker delta. The Lie superalgebra $\gl(\w{V})$ is spanned by $E_{i, j}$ with $i, j \in \hf\Z$.

The Lie superalgebra $\gl(\w{V})$ has a central extension, denoted by $\widehat{\gl}(\w{V})$, by the one-dimensional center $\C K$  corresponding to the following $2$-cocycle (cf. \cite[p. 99]{CL03}):
\begin{align}\label{central extension}
\tau(A, B):=\text{Str}([J,A]B),\qquad \hbox{$A,B\in \gl(\w{V})$,}
\end{align}
where $J=-\sum_{r\ge \hf} E_{r,r}$ and $\text{Str}$ denotes the supertrace. In fact, the cocycle $\tau$ is a coboundary. Moreover, there is an isomorphism $\iota$ from the direct sum $\gl(\w{V})\oplus \C K$ of Lie superalgebras to $\widehat{\gl}(\w{V})$ defined by
\begin{equation}\label{iso}
  \iota(A)=A+\mbox{Str}(JA)K,\quad\hbox{for $A\in\gl(\w{V})$, \quad and \quad $\iota(K)=K$.}
\end{equation}

\begin{rem}
The notion of central extensions is convenient for the formulation of truncation functors and super duality described in Section \ref{category-O}. The reader is encouraged to see \cite[Remark 3.3]{CLW11} for more explanations.
  \end{rem}

Let
\begin{align*}
{\mathbb J}_m(n) &=
\left\{\pm\hf, \pm \frac{3}{2}, \ldots, \pm(m-\hf) \right\}\cup\{0\}\cup \setc*{\pm j }{ j\in\N, j<n+1},\\
\ov{{\mathbb J}}_m(n) &=
\{ \pm{1},\ldots,\pm{m}\}\cup\{0\}\cup \setc[\Big]{\pm(j-\hf)}{j\in\N, j<n+1 },\\
\w{\mathbb J}(n)&=\setc[\Big]{r \in \hf \Z }{ -n \leq r \leq n }, \\
{{\mathbb J}}^\times_m(n) &={\mathbb J}_m(n)\backslash\{0\}, \qquad \ov{{\mathbb J}}^\times_m(n) = \ov{\mathbb J}_m(n) \backslash\{0\}, \qquad \w{\mathbb J}^\times(n)=\w{\mathbb J}(n))\backslash\{0\}, \\
{{\mathbb J}}^+_m(n)&=\setc*{r \in {\mathbb J}_m(n)}{r > 0 }, \\
\ov{{\mathbb J}}^+_m(n)&=\setc*{r \in \ov{\mathbb J}_m(n) }{ r > 0 },\\
\w{{\mathbb J}}^+(n)&=\setc*{r \in \w{{\mathbb J}}(n) }{ r > 0 }.
\end{align*}
We let $\wt{V}(n)$, $V_m(n)$, $\ov{V}_{\! m}(n)$, $\wt{V}^\times(n)$, $V^\times_m(n)$ and $\ov{V}^\times_{\! \! m}(n)$ be the subspaces
of $\wt{V}$ with ordered bases $\{v_i\}$ indexed by $\w{\mathbb J}(n)$, ${\mathbb J}_m(n)$, $\ov{{\mathbb J}}_m(n)$, $\w{\mathbb J}^\times(n)$, ${\mathbb J}^\times_m(n)$ and $\ov{{\mathbb J}}^\times_m(n)$, respectively.
This gives rise to subalgebras $\gl(\wt{V}(n))$, $\gl(V_m(n))$, $\gl(\ov{V}_{\!  m}(n))$, $\gl(\wt{V}^\times(n))$,
$\gl(V^\times_m(n))$ and $\gl(\ov{V}^\times_{\! \! m}(n))$ of the Lie superalgebra $\gl(\w{V})$. Let $\widehat{\bf b}:=\bigoplus_{r\le s, r,s\in \hf\Z}\C E_{r,s}$ denote the standard Borel subalgebra of $\gl(\w{V})$.

We will drop the symbol $(n)$ if $n=\infty$. For example, ${{\mathbb J}}_m:={{\mathbb J}}_m(\infty)$.
Define the total orders of ${\mathbb J}_m$ and $\ov{{\mathbb J}}_m$ by
\begin{align*}
\ldots&<_{{\mathbb J}_m}-3<_{{\mathbb J}_m}-2<_{{\mathbb J}_m}-1<_{{\mathbb J}_m}-(m-\hf)<_{{\mathbb J}_m}\ldots<_{{\mathbb J}_m}-\frac{3}{2}<_{{\mathbb J}_m}-\hf\\
&<_{{\mathbb J}_m}0<_{{\mathbb J}_m}\hf<_{{\mathbb J}_m}\frac{3}{2}<_{{\mathbb J}_m}\ldots<_{{\mathbb J}_m}{m}-\hf<_{{\mathbb J}_m}1<_{{\mathbb J}_m}2<_{{\mathbb J}_m}3<_{{\mathbb J}_m}\ldots
 \end{align*}
and
 \begin{align*}
\ldots&<_{\ov{\mathbb J}_m}-\frac{5}{2}<_{\ov{\mathbb J}_m}-\frac{3}{2}<_{\ov{\mathbb J}_m}-\hf<_{\ov{\mathbb J}_m}-m<_{\ov{\mathbb J}_m}\ldots<_{\ov{\mathbb J}_m}-2<_{\ov{\mathbb J}_m}-1\\
&<_{\ov{\mathbb J}_m}0<_{\ov{\mathbb J}_m}1<_{\ov{\mathbb J}_m}2<_{\ov{\mathbb J}_m}\ldots<_{\ov{\mathbb J}_m}{m}<_{\ov{\mathbb J}_m}\hf<_{\ov{\mathbb J}_m}\frac{3}{2}<_{\ov{\mathbb J}_m}\frac{5}{2}<_{\ov{\mathbb J}_m}\ldots,
 \end{align*}
respectively. The orderings give the standard Borel subalgebras
$$
{\bf b}[m]:=\bigoplus_{\substack{r\le_{{\mathbb J}_m} s,\\ r,s\in{{\mathbb J}_m}}}\C E_{r,s}\quad \hbox{ and}\quad \ov{{\bf b}}[m]:=\bigoplus_{\substack{r\le_{\ov{\mathbb J}_m} s,\\ r,s\in {\ov{\mathbb J}_m}}}\C E_{r,s}
$$
 of $\gl(V_m)$ and $\gl(\ov{V}_{\! m})$, respectively.

\subsection{General linear superalgebra $\wcG^{\mf{a}}$} \label{type-a}
Let $\w{V}^+(n)$, ${V}^+_m(n)$ and $\ov{V}^+_{\! \! m}(n)$ be the subspaces of $\w{V}$ with ordered bases $\{v_i\}$ indexed by $\w{\mathbb J}^+(n)$, ${\mathbb J}^+_m(n)$ and $\ov{\mathbb J}^+_m(n)$, respectively. Let $\wcG^{\mf{a}}_n$ denote the Lie subalgebra of $\gl(\wt{V})$ with basis $\{E_{i,j}\,|\, i,j\in \w{\J}^+(n)\}$. We denote $\wcG^{\mf{a}}:=\wcG^{\mf{a}}_\infty$.
Let $\cG^{\mf{a}}[m]_n$ and $\ovcG^{\mf{a}}[m]_n$ denote the Lie subalgebras of $\wcG^{\mf{a}}$
with bases $\{E_{i,j}\,|\, i,j\in {\J}^+_m(n)\}$ and $\{E_{i,j}\,|\, i,j\in \ov{\J}^+_m(n)\}$, respectively. The Lie (super)algebras $\cG^{\mf{a}}[m]_n$ and $\ovcG^{\mf{a}}[m]_n$ are isomorphic to $\gl(m | n)$. Let
$\w{\bf b}^{\mf{a}}:=\wcG^{\mf{a}}\cap\widehat{\bf b}$, ${\bf b}^{\mf{a}}[m]_n:=\cG^{\mf{a}}[m]_n \cap {\bf b}[m]$ and $\ov{\bf b}^{\mf{a}}[m]_n:=\ovcG^{\mf{a}}[m]_n \cap \ov{{\bf b}}[m]$ stand for the standard Borel subalgebras of $\wcG^{\mf{a}}$, $\cG^{\mf{a}}[m]_n$ and $\ovcG^{\mf{a}}[m]_n$, respectively. The corresponding Cartan subalgebras $\w{\bf h}^{\mf{a}}$, ${\bf h}^{\mf{a}}[m]_n$ and $\ov{\bf h}^{\mf{a}}[m]_n$ have bases $\{E^{\mf{a}}_i:=E_{i,i}\,|\, i\in \hf \N\}$, $\{E^{\mf{a}}_i \,|\, i\in {\J}^+_m(n)\}$ and $\{E^{\mf{a}}_i \,|\, i\in \ov{\J}^+_m(n)\}$, respectively. Let $\{\epsilon_i\}$ denote the dual bases of the Cartan subalgebras with the corresponding indices.

\subsection{Ortho-symplectic superalgebra $\wcG^{\mf{c}}$ and its subalgebras}\label{type-bd}\label{type-c}
Define a non-degenerate skew-supersymmetric bilinear form $(\cdot|\cdot)$ on $\w{V}^\times$ by
\begin{align}
(v_i |v_j)&=-(v_j |v_i)={\rm sgn}(i) \delta_{i, -j}, \quad & i, j \in \Z^*; \label{skew-sym-1}\\
(v_r |v_s)&=(v_s |v_r)=\delta_{r,-s}, \quad & r, s \in \hf +\Z; \label{skew-sym-2}\\
(v_i |v_r)&=(v_r |v_i)=0, \quad & i \in \Z^*, r \in \hf +\Z; \label{skew-sym-3}
\end{align}
where ${\rm sgn }(i)=1$ if $i >0$ and ${\rm sgn }(i)=-1$ if $i<0$.
The bilinear form induces non-degenerate bilinear forms on $V^\times_m(n)$ and $\ov{V}^\times_{\! \! m}(n)$.

Let $\wcG^{\mf{c}}_n$ (resp., $\cG^{\mf{c}}[m]_n$ and $\ovcG^{\mf{c}}[m]_n$) be the subalgebra of the Lie superalgebra $\gl(\w{V}^\times(n))$ (resp., $\gl({V}^\times_{m}(n))$ and $\gl(\ov{V}^\times_{\! \! m}(n))$) which preserves the bilinear form $(\cdot|\cdot)$. The Lie superalgebra $\wcG^{\mf{c}}:=\wcG^{\mf{c}}_\infty$ is spanned by the following elements ($i, j \in \Z^*$ and $r, s \in \hf+\Z$):
\begin{align*}
E^{\mf{c}}_{i, j} &:=-E^{\mf{c}}_{-j, -i}:=E_{i,j}-E_{-j, -i}, \quad & ij>0; \\
E^{\mf{c}}_{i, j} &:=E^{\mf{c}}_{-j, -i}:=E_{i,j}+E_{-j, -i}, \quad & ij<0; \\
E^{\mf{c}}_{r, s}&:=-E^{\mf{c}}_{-s, -r}:=E_{r,s}- E_{-s, -r}; \\
E^{\mf{c}}_{i, r} &:= E^{\mf{c}}_{-r, -i}:=E_{i,r}+ E_{-r, -i}, \quad & i>0; \\
E^{\mf{c}}_{i, r} &:= -E^{\mf{c}}_{-r, -i}:=E_{i,r}- E_{-r, -i}, \quad & i<0.
\end{align*}
The subalgebras $\cG^{\mf{c}}[m]_n$ and ${\ovcG}^{\mf{c}}[m]_n$ of $\wcG^{\mf{c}}$ are spanned by $E^{\mf{c}}_{i, j}$ with $i, j \in \J^\times_m(n)$ and $\ov\J^\times_m(n)$, respectively. Note that $\cG^{\mf{c}}[m]_n$ is isomorphic to $\osp(2m | 2n)$ while $\ovcG^{\mf{c}}[m]_n$ is isomorphic to $\spo(2m | 2n)$.

We let
$\w{\bf b}^{\mf{c}}:=\wcG^{\mf{c}} \cap\widehat{\bf b}$, ${\bf b}^{\mf{c}}[m]_n:=\cG^{\mf{c}}[m]_n \cap {\bf b}[m]$ and $\ov{\bf b}^{\mf{c}}[m]_n:=\ovcG^{\mf{c}}[m]_n \cap \ov{{\bf b}}[m]$ stand for the standard Borel subalgebras of $\wcG^{\mf{c}}$, $\cG^{\mf{c}}[m]_n$ and $\ovcG^{\mf{c}}[m]_n$, respectively. The corresponding Cartan subalgebras $\w{\bf h}^{\mf{c}}$, ${\bf h}^{\mf{c}}[m]_n$ and $\ov{\bf h}^{\mf{c}}[m]_n$ have bases $\{E^\mf{c}_i:=E_{i,i}-E_{-i,-i} \,|\, i\in \hf \N \}$, $\{E^\mf{c}_i\,|\, i\in {\J}^+_m(n)\}$ and $\{E^\mf{c}_i\,|\, i\in \ov{\J}^+_m(n)\}$, respectively. Let $\{\epsilon_i\}$ denote the dual bases of the Cartan subalgebras with the corresponding indices.

\subsection{Ortho-symplectic superalgebra $\wcG^{\mf{d}}$ and its subalgebras}\label{type-d}
Define a non-degenerate supersymmetric bilinear form $(\cdot|\cdot)$ on $\w{V}^\times$ by
\begin{align}
(v_i |v_j)&=(v_j |v_i)= \delta_{i, -j}, \quad & i, j \in \Z^*; \label{sym-1}\\
(v_r |v_s)&=-(v_s |v_r)={\rm sgn}(r)\delta_{r,-s}, \quad & r, s \in \hf +\Z; \label{sym-2}\\
(v_i |v_r)&=(v_r |v_i)=0, \quad & i \in \Z^*, r \in \hf +\Z. \label{sym-3}
\end{align}
The bilinear form induces non-degenerate bilinear forms on $V^\times_{m}(n)$ and $\ov{V}^\times_{\! \! m}(n)$.

Let $\wcG^{\mf{d}}_n$ (resp., $\cG^{\mf{d}}[m]_n$ and ${\ovcG}^{\mf{d}}[m]_n$) be the subalgebra of the Lie superalgebra $\gl(\w{V}^\times(n))$ (resp., $\gl({V}^\times_{m}(n))$ and $\gl(\ov{V}^\times_{\! \! m}(n))$) which preserves the bilinear form $(\cdot|\cdot)$.
The Lie superalgebra $\wcG^{\mf{d}}:=\wcG^{\mf{d}}_\infty$ is spanned by the following elements ($i, j \in \Z^*$ and $r, s \in \hf+\Z$):
\begin{align*}
E^{\mf{d}}_{i, j} &:=-E^{\mf{d}}_{-j, -i}:=E_{i, j}-E_{-j, -i}; \\
E^{\mf{d}}_{r, s}&:=-E^{\mf{d}}_{-s, -r}:=E_{r, s}- E_{-s, -r},  \quad & rs>0; \\
E^{\mf{d}}_{r, s}&:=E^{\mf{d}}_{-s, -r}:=E_{r, s}+ E_{-s, -r},  \quad & rs <0; \\
E^{\mf{d}}_{i, r} &:= E^{\mf{d}}_{-r, -i}:=E_{i, r}+ E_{-r, -i}, \quad & r>0; \\
E^{\mf{d}}_{i, r} &:= -E^{\mf{d}}_{-r, -i}:=E_{i, r}- E_{-r, -i}, \quad & r<0.
\end{align*}
The subalgebras $\cG^{\mf{d}}[m]_n$ and ${\ovcG}^{\mf{d}}[m]_n$ of $\wcG^{\mf{d}}$ are spanned by $E^{\mf{d}}_{i, j}$ with $i, j \in \J^\times_m(n)$ and $\ov\J^\times_m(n)$, respectively. Note that $\cG^{\mf{d}}[m]_n$ is isomorphic to $\spo(2m | 2n)$ while $\ovcG^{\mf{d}}[m]_n$ is isomorphic to $\osp(2m | 2n)$.

We let
$\w{\bf b}^{\mf{d}}:=\wcG^{\mf{d}} \cap\widehat{\bf b}$, ${\bf b}^{\mf{d}}[m]_n:=\cG^{\mf{d}}[m]_n \cap {\bf b}[m]$ and $\ov{\bf b}^{\mf{d}}[m]_n:=\ovcG^{\mf{d}}[m]_n \cap \ov{{\bf b}}[m]$ stand for the standard Borel subalgebras of $\wcG^{\mf{d}}$, $\cG^{\mf{d}}[m]_n$ and $\ovcG^{\mf{d}}[m]_n$, respectively. The corresponding Cartan subalgebras $\w{\bf h}^{\mf{d}}$, ${\bf h}^{\mf{d}}[m]_n$ and $\ov{\bf h}^{\mf{d}}[m]_n$ have bases $\{E^\mf{d}_i:=E_{i,i}-E_{-i,-i} \,|\, i\in \hf \N \}$, $\{E^\mf{d}_i\,|\, i\in {\J}^+_m(n)\}$ and $\{E^\mf{d}_i\,|\, i\in \ov{\J}^+_m(n)\}$, respectively. Let $\{\epsilon_i\}$ denote the dual bases of the Cartan subalgebras with the corresponding indices.

Define a linear automorphism $\varphi$ of degree $\ov 1$  on the superspace $\w V^\times$ by
\begin{equation}\label{phi}
  \varphi(v_{\pm r}):=\left\{
                       \begin{array}{ll}
                          v_{\pm (r-\hf)}, & \hbox{if $r\in \N$;} \\
                          v_{\pm (r+\hf)}, & \hbox{if $r\in \hf+\Z_+$.}
                       \end{array}
                     \right.
\end{equation}
The automorphism $\varphi$ induces an automorphism $\widehat{\varphi}$ on the Lie superalgebra $\gl(\wt{V}^\times)$. The supersymmetric bilinear form on $\w V^\times$ defined by \eqref{sym-1}, \eqref{sym-2} and \eqref{sym-3} is exactly the bilinear form induced, via ${\varphi}$, by the skew-supersymmetric bilinear form on $\w V^\times$ defined by \eqref{skew-sym-1}, \eqref{skew-sym-2} and \eqref{skew-sym-3}.
The restriction of $\widehat{\varphi}$ to $\wcG^{\mf{c}}$ gives an isomorphism from $\wcG^{\mf{c}}$ to $\wcG^{\mf{d}}$ and hence an isomorphism from $\cG^{\mf{c}}[m]_n$ (resp., $\ovcG^{\mf{c}}[m]_n$) to $\ovcG^{\mf{d}}[m]_n$ (resp., $\cG^{\mf{d}}[m]_n$). It is evident that $\widehat{\varphi}$ preserves the corresponding Borel and Cartan subalgebras. The restrictions of $\widehat{\varphi}$ are denoted by $\widehat{\varphi}$ as well.
We summarize the results in the following lemma.

\begin{lem}\label{lem:phi}
There is an isomorphism $\widehat{\varphi}$ from $\wcG^{\mf{c}}$ to $\wcG^{\mf{d}}$ given by
\begin{equation}\label{phi_cd}
\widehat{\varphi}(E^{\mf{c}}_{r, s})=
\begin{cases}
E^{\mf{d}}_{r-\hf, s-\hf},\quad\text{if } r,s\in \N\cup(-(\hf+\Z_+));\\
E^{\mf{d}}_{r+\hf, s+\hf},\quad\text{if } r,s\in -\N\cup(\hf+\Z_+);\\
E^{\mf{d}}_{r-\hf, s+\hf},\quad\text{if } r\in \N\cup(-(\hf+\Z_+)),s\in -\N\cup(\hf+\Z_+);\\
E^{\mf{d}}_{r+\hf, s-\hf},\quad\text{if } r\in -\N\cup(\hf+\Z_+), s\in \N\cup(-(\hf+\Z_+)).
\end{cases}
\end{equation}
The restrictions of $\widehat{\varphi}$ to $\cG^{\mf{c}}[m]_n$ and $\ovcG^{\mf{c}}[m]_n$ give isomorphisms $\widehat{\varphi}: \cG^{\mf{c}}[m]_n\longrightarrow \ovcG^{\mf{d}}[m]_n$ and $\widehat{\varphi}: \ovcG^{\mf{c}}[m]_n \longrightarrow\cG^{\mf{d}}[m]_n$, respectively.
\end{lem}

\subsection{Dynkin diagrams} \label{Dynkin}
\label{dynkin}
Consider the free abelian group with basis
$\{\ep_r \, | \, r \in \hf\N \}$. It is endowed with a symmetric bilinear form $(\cdot,\cdot)$ defined by
$$
(\ep_r,\ep_s)=(-1)^{2r}\delta_{r,s}, \qquad r,s \in \hf\N.
$$
The parity of $\ep_r$ is defined as follows: $|\ep_r|=0$ for $r \in \N$ and $|\ep_r|=1$ for $r \in \frac{1}{2}+\Z_+$.
Let
$$
\alpha_{\times}=\ep_{m}-\ep_{\frac{1}{2}}, \quad \alpha_{r}=\ep_{r}-\ep_{r+\frac{1}{2}},\quad \beta_{r}=\ep_{r}-\ep_{r+1},\quad r\in \hf\N.
$$

For $\xx =\mf{a, c, d}$, the Dynkin diagrams of the Lie superalgebras $\wcG^{\xx}_n$ and $\ovcG^{\xx}[m]_n$ (where $m \in \N$) together with prescribed fundamental systems are listed below (cf. \cite[Section 2.5]{K}). In what follows, $\bigcirc$ and $\bigotimes$ denote an even simple root and an odd isotropic simple root, respectively.

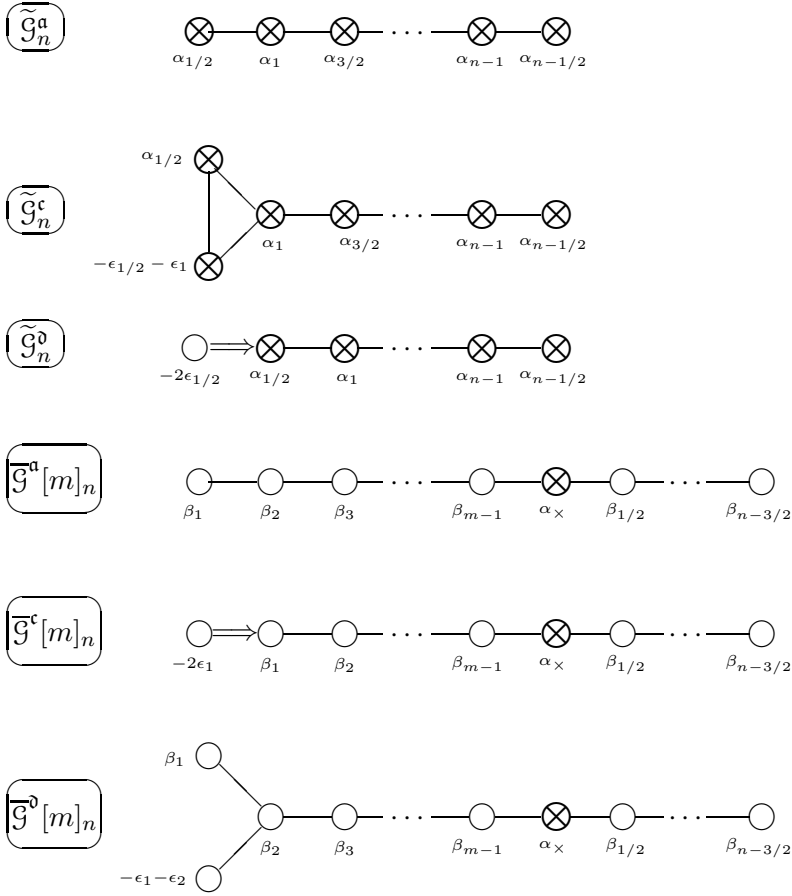
\begin{center}
\hskip -3cm \setlength{\unitlength}{0.16in} \medskip
\begin{picture}(24,3)
\put(-.5,1.4){{\ovalBox(1.8,1.5){$\wcG^{\mf a}_n$}}}

\put(5.7,2){\makebox(0,0)[c]{$\bigotimes$}}
\put(8,2){\makebox(0,0)[c]{$\bigotimes$}}
\put(10.4,2){\makebox(0,0)[c]{$\bigotimes$}}
\put(14.85,2){\makebox(0,0)[c]{$\bigotimes$}}
\put(17.25,2){\makebox(0,0)[c]{$\bigotimes$}}

\put(6,2){\line(1,0){1.55}}
\put(8.4,2){\line(1,0){1.55}} \put(10.82,2){\line(1,0){0.8}}
\put(13.2,2){\line(1,0){1.2}} \put(15.28,2){\line(1,0){1.45}}

\put(12.5,1.95){\makebox(0,0)[c]{$\cdots$}}

\put(5.5,1){\makebox(0,0)[c]{\tiny $\alpha_{1/2}$}}
\put(8,1){\makebox(0,0)[c]{\tiny $\alpha_{1}$}}
\put(10.4,1){\makebox(0,0)[c]{\tiny $\alpha_{3/2}$}}
\put(14.8,1){\makebox(0,0)[c]{\tiny $\alpha_{n-1}$}}
\put(17.15,1){\makebox(0,0)[c]{\tiny $\alpha_{n-1/2}$}}
\end{picture}
\end{center}

\begin{center}
\hskip -3cm \setlength{\unitlength}{0.16in}
\begin{picture}(24,4)

\put(-.5,1.4){{\ovalBox(1.8,1.5){$\wcG^{\mf c}_n$}}}

\put(8,2){\makebox(0,0)[c]{$\bigotimes$}}
\put(10.4,2){\makebox(0,0)[c]{$\bigotimes$}}
\put(14.85,2){\makebox(0,0)[c]{$\bigotimes$}}
\put(17.25,2){\makebox(0,0)[c]{$\bigotimes$}}

\put(6,3.8){\makebox(0,0)[c]{$\bigotimes$}}
\put(6,.3){\makebox(0,0)[c]{$\bigotimes$}}
\put(8.4,2){\line(1,0){1.55}} \put(10.82,2){\line(1,0){0.8}}
\put(13.2,2){\line(1,0){1.2}} \put(15.28,2){\line(1,0){1.45}}

\put(12.5,1.95){\makebox(0,0)[c]{$\cdots$}}

\put(7.5,2.2){\line(-1,1){1.3}}
\put(7.6,1.8){\line(-1,-1){1.25}}
\put(6,.8){\line(0,1){2.6}}

\put(4.5,3.8){\makebox(0,0)[c]{\tiny $\alpha_{1/2}$}}
\put(3.8,0.3){\makebox(0,0)[c]{\tiny $-\epsilon_{1/2}-\epsilon_{1}$}}

\put(8.1,1){\makebox(0,0)[c]{\tiny $\alpha_{1}$}}
\put(10.9,1){\makebox(0,0)[c]{\tiny $\alpha_{3/2}$}}
\put(14.8,1){\makebox(0,0)[c]{\tiny $\alpha_{n-1}$}}
\put(17.15,1){\makebox(0,0)[c]{\tiny $\alpha_{n-1/2}$}}

\end{picture}
\end{center}

\begin{center}
\hskip -3cm \setlength{\unitlength}{0.16in}
\begin{picture}(24,3)
\put(-.5,1.4){{\ovalBox(1.8,1.5){$\wcG^{\mf d}_n$}}}

\put(5.55,2){\makebox(0,0)[c]{$\bigcirc$}}
\put(8,2){\makebox(0,0)[c]{$\bigotimes$}}
\put(10.4,2){\makebox(0,0)[c]{$\bigotimes$}}
\put(14.85,2){\makebox(0,0)[c]{$\bigotimes$}}
\put(17.25,2){\makebox(0,0)[c]{$\bigotimes$}}

\put(8.4,2){\line(1,0){1.55}} \put(10.82,2){\line(1,0){0.8}}
\put(13.2,2){\line(1,0){1.2}} \put(15.28,2){\line(1,0){1.45}}

\put(12.5,1.95){\makebox(0,0)[c]{$\cdots$}}

\put(6,1.8){$\Longrightarrow$}

\put(5.4,1){\makebox(0,0)[c]{\tiny $-2\epsilon_{1/2}$}}

\put(8,1){\makebox(0,0)[c]{\tiny $\alpha_{1/2}$}}
\put(10.5,1){\makebox(0,0)[c]{\tiny $\alpha_{1}$}}
\put(14.8,1){\makebox(0,0)[c]{\tiny $\alpha_{n-1}$}}
\put(17.15,1){\makebox(0,0)[c]{\tiny $\alpha_{n-1/2}$}}

\end{picture}
\end{center}


\begin{center}
\hskip -3cm \setlength{\unitlength}{0.16in} \medskip
\begin{picture}(24,3)
\put(-.5,1){{\ovalBox(3,2,2){$\ovcG^{\mf a}[m]_n$}}}

\put(5.7,2){\makebox(0,0)[c]{$\bigcirc$}}
\put(8,2){\makebox(0,0)[c]{$\bigcirc$}}
\put(10.4,2){\makebox(0,0)[c]{$\bigcirc$}}
\put(14.85,2){\makebox(0,0)[c]{$\bigcirc$}}
\put(17.25,2){\makebox(0,0)[c]{$\bigotimes$}}
\put(19.4,2){\makebox(0,0)[c]{$\bigcirc$}}
\put(23.9,2){\makebox(0,0)[c]{$\bigcirc$}}

\put(6,2){\line(1,0){1.55}}
\put(8.4,2){\line(1,0){1.55}} \put(10.82,2){\line(1,0){0.8}}
\put(13.2,2){\line(1,0){1.2}} \put(15.28,2){\line(1,0){1.45}}
\put(17.7,2){\line(1,0){1.25}} \put(19.8,2){\line(1,0){0.9}}
\put(22.1,2){\line(1,0){1.4}}

\put(12.5,1.95){\makebox(0,0)[c]{$\cdots$}}
\put(21.5,1.95){\makebox(0,0)[c]{$\cdots$}}

\put(5.5,1){\makebox(0,0)[c]{\tiny $\beta_1$}}
\put(8,1){\makebox(0,0)[c]{\tiny $\beta_2$}}
\put(10.4,1){\makebox(0,0)[c]{\tiny $\beta_3$}}
\put(14.7,1){\makebox(0,0)[c]{\tiny $\beta_{m-1}$}}
\put(17.15,1){\makebox(0,0)[c]{\tiny $\alpha_{\times}$}}
\put(19.5,1){\makebox(0,0)[c]{\tiny $\beta_{1/2}$}}
\put(23.8,1){\makebox(0,0)[c]{\tiny $\beta_{n-3/2}$}}
\end{picture}
\end{center}

\begin{center}
\hskip -3cm \setlength{\unitlength}{0.16in} \medskip
\begin{picture}(24,3)

\put(-.5,1){{\ovalBox(3,2,2){$\ovcG^{\mf c}[m]_n$}}}

\put(5.7,2){\makebox(0,0)[c]{$\bigcirc$}}
\put(8,2){\makebox(0,0)[c]{$\bigcirc$}}
\put(10.4,2){\makebox(0,0)[c]{$\bigcirc$}}
\put(14.85,2){\makebox(0,0)[c]{$\bigcirc$}}
\put(17.25,2){\makebox(0,0)[c]{$\bigotimes$}}
\put(19.4,2){\makebox(0,0)[c]{$\bigcirc$}}
\put(23.9,2){\makebox(0,0)[c]{$\bigcirc$}}

\put(6.85,2){\makebox(0,0)[c]{$\Longrightarrow$}}

\put(8.4,2){\line(1,0){1.55}} \put(10.82,2){\line(1,0){0.8}}
\put(13.2,2){\line(1,0){1.2}} \put(15.28,2){\line(1,0){1.45}}
\put(17.7,2){\line(1,0){1.25}} \put(19.8,2){\line(1,0){0.9}}
\put(22.1,2){\line(1,0){1.4}}

\put(12.5,1.95){\makebox(0,0)[c]{$\cdots$}}
\put(21.5,1.95){\makebox(0,0)[c]{$\cdots$}}

\put(5.5,1){\makebox(0,0)[c]{\tiny$-2\ep_{1}$}}
\put(8,1){\makebox(0,0)[c]{\tiny $\beta_1$}}
\put(10.4,1){\makebox(0,0)[c]{\tiny $\beta_2$}}
\put(14.7,1){\makebox(0,0)[c]{\tiny $\beta_{m-1}$}}
\put(17.15,1){\makebox(0,0)[c]{\tiny $\alpha_{\times}$}}
\put(19.5,1){\makebox(0,0)[c]{\tiny $\beta_{1/2}$}}
\put(23.8,1){\makebox(0,0)[c]{\tiny $\beta_{n-3/2}$}}
\end{picture}
\end{center}

\begin{center}
\hskip -3cm \setlength{\unitlength}{0.16in} \medskip
\begin{picture}(24,4)

\put(-.5,1){{\ovalBox(3,2,2){$\ovcG^{\mf d}[m]_n$}}}

\put(6,0){\makebox(0,0)[c]{$\bigcirc$}}
\put(6,4){\makebox(0,0)[c]{$\bigcirc$}}
\put(8,2){\makebox(0,0)[c]{$\bigcirc$}}
\put(10.4,2){\makebox(0,0)[c]{$\bigcirc$}}
\put(14.85,2){\makebox(0,0)[c]{$\bigcirc$}}
\put(17.25,2){\makebox(0,0)[c]{$\bigotimes$}}
\put(19.4,2){\makebox(0,0)[c]{$\bigcirc$}}
\put(23.9,2){\makebox(0,0)[c]{$\bigcirc$}}

\put(6.35,0.3){\line(1,1){1.35}} \put(6.35,3.7){\line(1,-1){1.35}}
\put(8.4,2){\line(1,0){1.55}} \put(10.82,2){\line(1,0){0.8}}

\put(13.2,2){\line(1,0){1.2}} \put(15.28,2){\line(1,0){1.45}}
\put(17.7,2){\line(1,0){1.25}} \put(19.8,2){\line(1,0){0.9}}
\put(22.1,2){\line(1,0){1.4}}

\put(12.5,1.95){\makebox(0,0)[c]{$\cdots$}}
\put(21.5,1.95){\makebox(0,0)[c]{$\cdots$}}

\put(4.9,3.9){\makebox(0,0)[c]{\tiny$\beta_{1}$}}
\put(4.2,0){\makebox(0,0)[c]{\tiny${-}\ep_{1}{-}\ep_{2}$}}

\put(8,1){\makebox(0,0)[c]{\tiny $\beta_2$}}
\put(10.4,1){\makebox(0,0)[c]{\tiny $\beta_3$}}
\put(14.7,1){\makebox(0,0)[c]{\tiny $\beta_{m-1}$}}
\put(17.15,1){\makebox(0,0)[c]{\tiny $\alpha_{\times}$}}
\put(19.5,1){\makebox(0,0)[c]{\tiny $\beta_{1/2}$}}
\put(23.8,1){\makebox(0,0)[c]{\tiny $\beta_{n-3/2}$}}
\end{picture}
\end{center}

The Dynkin diagram of the Lie algebra $\ovcG^{\mf{a}}[0]_n$ is given below:

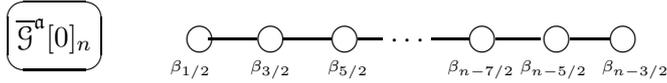
\begin{center}
\hskip -3cm \setlength{\unitlength}{0.16in}
\begin{picture}(24,3)

\put(-.5,1){{\ovalBox(3,2,2){$\ovcG^{\mf a}[0]_n$}}}

\put(5.7,2){\makebox(0,0)[c]{$\bigcirc$}}
\put(8,2){\makebox(0,0)[c]{$\bigcirc$}}
\put(10.4,2){\makebox(0,0)[c]{$\bigcirc$}}
\put(14.85,2){\makebox(0,0)[c]{$\bigcirc$}}
\put(17.25,2){\makebox(0,0)[c]{$\bigcirc$}}
\put(19.4,2){\makebox(0,0)[c]{$\bigcirc$}}

\put(6,2){\line(1,0){1.55}}
\put(8.4,2){\line(1,0){1.55}} \put(10.82,2){\line(1,0){0.8}}
\put(13.2,2){\line(1,0){1.2}} \put(15.28,2){\line(1,0){1.45}}
\put(17.7,2){\line(1,0){1.25}}

\put(12.5,1.95){\makebox(0,0)[c]{$\cdots$}}

\put(5.4,1){\makebox(0,0)[c]{\tiny $\beta_{1/2}$}}
\put(8,1){\makebox(0,0)[c]{\tiny $\beta_{3/2}$}}
\put(10.5,1){\makebox(0,0)[c]{\tiny $\beta_{5/2}$}}
\put(14.8,1){\makebox(0,0)[c]{\tiny $\beta_{n-7/2}$}}
\put(17.15,1){\makebox(0,0)[c]{\tiny $\beta_{n-5/2}$}}
\put(19.8,1){\makebox(0,0)[c]{\tiny $\beta_{n-3/2}$}}

\end{picture}
\end{center}
For $\xx=\mf{c,d}$, the Dynkin diagrams of the Lie algebras $\ovcG^{\xx}[0]_n$ are somewhat different from the ones of $\ovcG^{\xx}[m]_n$, where $m \in \N$.
In fact, $\ovcG^{\mf{c}}[0]_n \cong \cG^{\mf{d}}[0]_n  \cong \mf{so}(2n)$ and $\ovcG^{\mf{d}}[0]_n \cong \cG^{\mf{c}}[0]_n   \cong \mf{sp}(2n)$.

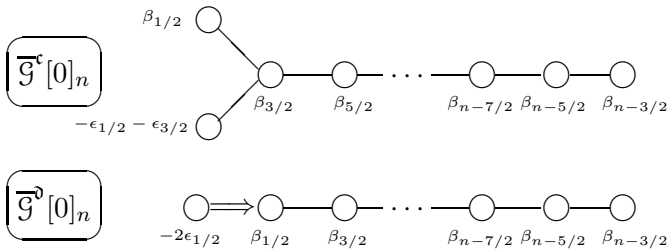
\begin{center}
\hskip -3cm \setlength{\unitlength}{0.16in}
\begin{picture}(24,4)

\put(-.5,1){{\ovalBox(3,2,2){$\ovcG^{\mf c}[0]_n$}}}

\put(8,2){\makebox(0,0)[c]{$\bigcirc$}}
\put(10.4,2){\makebox(0,0)[c]{$\bigcirc$}}
\put(14.85,2){\makebox(0,0)[c]{$\bigcirc$}}
\put(17.25,2){\makebox(0,0)[c]{$\bigcirc$}}
\put(19.4,2){\makebox(0,0)[c]{$\bigcirc$}}
\put(6,3.8){\makebox(0,0)[c]{$\bigcirc$}}
\put(6,.3){\makebox(0,0)[c]{$\bigcirc$}}

\put(8.4,2){\line(1,0){1.55}} \put(10.82,2){\line(1,0){0.8}}
\put(13.2,2){\line(1,0){1.2}} \put(15.28,2){\line(1,0){1.45}}
\put(17.7,2){\line(1,0){1.25}}
\put(7.6,2.2){\line(-1,1){1.3}}
\put(7.6,1.8){\line(-1,-1){1.3}}

\put(12.5,1.95){\makebox(0,0)[c]{$\cdots$}}

\put(4.5,3.8){\makebox(0,0)[c]{\tiny $\beta_{1/2}$}}
\put(3.5,0.3){\makebox(0,0)[c]{\tiny $-\epsilon_{1/2}-\epsilon_{3/2}$}}

\put(8.1,1){\makebox(0,0)[c]{\tiny $\beta_{3/2}$}}
\put(10.7,1){\makebox(0,0)[c]{\tiny $\beta_{5/2}$}}
\put(14.8,1){\makebox(0,0)[c]{\tiny $\beta_{n-7/2}$}}
\put(17.15,1){\makebox(0,0)[c]{\tiny $\beta_{n-5/2}$}}
\put(19.8,1){\makebox(0,0)[c]{\tiny $\beta_{n-3/2}$}}

\end{picture}
\end{center}

\begin{center}
\hskip -3cm \setlength{\unitlength}{0.16in}
\begin{picture}(24,3)

\put(-.5,1){{\ovalBox(3,2,2){$\ovcG^{\mf d}[0]_n$}}}
\put(8,2){\makebox(0,0)[c]{$\bigcirc$}}
\put(10.4,2){\makebox(0,0)[c]{$\bigcirc$}}
\put(14.85,2){\makebox(0,0)[c]{$\bigcirc$}}
\put(17.25,2){\makebox(0,0)[c]{$\bigcirc$}}
\put(19.4,2){\makebox(0,0)[c]{$\bigcirc$}}
\put(5.6,2){\makebox(0,0)[c]{$\bigcirc$}}
\put(8.4,2){\line(1,0){1.55}} \put(10.82,2){\line(1,0){0.8}}
\put(13.2,2){\line(1,0){1.2}} \put(15.28,2){\line(1,0){1.45}}
\put(17.7,2){\line(1,0){1.25}}
\put(6,1.8){$\Longrightarrow$}
\put(12.5,1.95){\makebox(0,0)[c]{$\cdots$}}

\put(5.4,1){\makebox(0,0)[c]{\tiny $-2\epsilon_{1/2}$}}
\put(8,1){\makebox(0,0)[c]{\tiny $\beta_{1/2}$}}
\put(10.5,1){\makebox(0,0)[c]{\tiny $\beta_{3/2}$}}
\put(14.8,1){\makebox(0,0)[c]{\tiny $\beta_{n-7/2}$}}
\put(17.15,1){\makebox(0,0)[c]{\tiny $\beta_{n-5/2}$}}
\put(19.8,1){\makebox(0,0)[c]{\tiny $\beta_{n-3/2}$}}

\end{picture}
\end{center}

From here on, let $\w{\Phi}^{\xx\, +}_n$ (resp., $\Phi^{\xx}[m]_n^+$ and $\ov{\Phi}^{\xx}[m]^+_n$) denote the set of positive roots of $\wcG^{\xx}_n$ (resp., $\cG^{\xx}[m]_n$ and $\ovcG^{\xx}[m]_n$).

\subsection{Central extensions}\label{Central Ext}
For $\xx=\mf{a, c, d}$, consider the central extension $\DG^{\xx}$ (resp., $\G^{\xx}[m]_n$ and $\SG^{\xx}[m]_n$) of $\wcG^{\xx}$ (resp., $\cG^{\xx}[m]_n$ and $\ovcG^{\xx}[m]_n$) by the one-dimensional center $\C K$, which is inherited from the central extension $\widehat\gl(\w{V})$ of ${\gl}(\w{V})$ determined by the $2$-cocycle given in \eqref{central extension}.
The restriction of the isomorphism $\iota$ to $\wcG^\xx \oplus \C K$ (resp., $\cG^{\xx}[m]_n \oplus \C K$ and $\ovcG^{\xx}[m]_n \oplus \C K$) is an isomorphism $\iota: \wcG^\xx \oplus \C K \to \DG^{\xx}$ (resp., $\iota:\cG^{\xx}[m]_n \oplus \C K \to \G^{\xx}[m]_n$ and $\iota: \ovcG^{\xx}[m]_n \oplus \C K \to \SG^{\xx}[m]_n$) given by
\begin{equation}\label{iso-e}
\iota(A)=A+\mbox{Str}(JA)K \quad {\rm and} \quad \iota(K)=K
\end{equation}
for $A \in \wcG^\xx$ (resp., $\cG^{\xx}[m]_n$ and $\ovcG^{\xx}[m]_n$).

Note that $[J, A]=0$ for all $A\in \wcG^{\mf a}$. Thus the $2$-cocycle is zero when it restricts to $\wcG^{\mf a} \times \wcG^{\mf a}$, and hence
$$
[A, B]_{\DG^{\mf a}}=[A,B],\qquad \hbox{for $A, B\in \wcG^{\mf a}$,}
$$
where $[\cdot, \cdot]_{\DG^{\mf a}}$ denotes the Lie bracket on $\DG^{\mf a}$. For notational unity, we still consider $\DG^{\mf a}$, $\G^{\mf a}[m]_n$ and $\SG^{\mf a}[m]_n$.

\begin{rem}
  Every $\DG^{\xx}$(resp., $\G^{\xx}[m]_n$ and $\SG^{\xx}[m]_n$)-module can be regarded as a $\wcG^{\xx}$(resp., $\cG^{\xx}[m]_n$ and $\ovcG^{\xx}[m]_n$)-module via the isomorphism \eqnref{iso-e}.
\end{rem}

We let
$\w{\mf b}^{\xx}:=\w{\bf b}^{\xx}\oplus \C K$, ${\mf b}^{\xx}[m]_n:={\bf b}^{\xx}[m]_n\oplus \C K$ and $\ov{\mf b}^{\xx}[m]_n:=\ov{\bf b}^{\xx}[m]_n \oplus \C K$ stand for the standard Borel subalgebras of $\DG^{\xx}$, $\G^{\xx}[m]_n$ and $\SG^{\xx}[m]_n$, respectively, and let $\wfh^{\xx}$, $\fh^{\xx}[m]_n$ and $\ovfh^{\xx}[m]_n$ denote the Cartan subalgebras of $\DG^{\xx}$, $\G^{\xx}[m]_n$ and $\SG^{\xx}[m]_n$ spanned by bases $\{K, E^\xx_r \}$ with dual bases $\{\Lambda_0,\ep_r\}$ in the restricted dual $\wfh^{\xx \, *}$, $\fh^{\xx}[m]_n^*$ and $\ovfh^{\xx}[m]_n^*$, where $r$ runs over the index sets $\hf \N$, ${\J}^+_m(n)$ and $\ov{\J}^+_m(n)$, respectively.
Here $\La_0$ is the element of $\wfh^{\xx \, *}$ (resp., $\fh^{\xx}[m]_n^*$ and $\ovfh^{\xx}[m]_n^*$) defined by
 $$
 \La_0(K)=1 \,\,\, \textrm{ and }   \,\,\,\La_0(E^\xx_r)=0
 $$
  for all $r \in \hf \N$ (resp., ${\J}^+_m(n)$ and $\ov{\J}^+_m(n)$).

It is easy to see that the automorphism $\widehat\varphi$ on $\gl(\w V^\times)$ induced by $\varphi$ defined in \eqref{phi} can be extended to the central extension $\widehat\gl(\w V^\times)$ of $\gl(\w V^\times)$. By \lemref{lem:phi}, we have the following lemma.

\begin{lem}\label{lem:phi_ext}
The isomorphism $\widehat{\varphi}: \wcG^{\mf{c}}\longrightarrow\wcG^{\mf{d}}$ extends to an isomorphism from
$\DG^{\mf{c}}$ to $\DG^{\mf{d}}$, which is also denoted by $\widehat{\varphi}$ and is given by \eqref{phi_cd} together with
$$
\widehat{\varphi}(K)=-K.
$$
The restrictions of $\widehat{\varphi}$ to $\G^{\mf{c}}[m]_n$ and $\ov\G^{\mf{c}}[m]_n$ give isomorphisms
$\widehat{\varphi}: \G^{\mf{c}}[m]_n\longrightarrow \ov\G^{\mf{d}}[m]_n$
and $\widehat{\varphi}: \ov\G^{\mf{c}}[m]_n \longrightarrow\G^{\mf{d}}[m]_n$, respectively.
\end{lem}

\subsection{Parabolic BGG categories and super duality} \label{category-O}
For details on the materials in this subsection, we refer the readers to \cite[Sections 2 and 3]{CL10} for type $\mf{a}$ and \cite[Sections 2 and 3]{CLW11} for types $\mf{c,d}$ (see also \cite[Sections 6.1 and 6.2]{CW} and \cite[Section 2.4]{CaL}).

Let $\w{\mf l}^\xx$, $\mf{l}^\xx[m]_n$ and $\ov{\mf{l}}^\xx[m]_n$ be the Levi subalgebras of $\DG^\xx$, $\G^\xx[m]_n$ and $\SG^\xx[m]_n$ defined by
$$
\w{\mf l}^\xx:=\DG^\xx \cap \widehat{\bf l},
\quad \mf{l}^\xx[m]_n:=\G^\xx[m]_n \cap \widehat{\bf l},
\quad\hbox{and}
\quad \ov{\mf l}^\xx[m]_n:=\SG^\xx[m]_n\cap\widehat{\bf l},
$$
respectively, and let $\w{\mf p}^\xx=\tilde{\mf l}^\xx+\w{\mf b}^\xx$, $\mf{p}^\xx[m]_n=\mf{l}^\xx[m]_n+{\mf b}^\xx[m]_n$ and $\ov{\mf{p}}^\xx[m]_n=\ov{\mf{l}}^\xx[m]_n+\ov{\mf b}^\xx[m]_n$ be the corresponding parabolic subalgebras, where $\widehat{\bf l} :=\bigoplus_{\substack{rs>0,\\ r,s\in \hf\Z^*}}\C E_{r, s}\oplus\C K$. Observe that
$$
{\mf l}^\xx[m]_n\cong \G^{\mf a}[m]_n\cong\gl(m|n)\oplus\C K\qquad\hbox{and }\qquad \ov{\mf l}^\xx[m]_n\cong\SG^{\mf a}[m]_n\cong\gl(m|n)\oplus\C K.
$$
For $\xx={\mf a}$, we have $\w{\mf l}^{\mf a}=\DG^{\mf a}$, ${\mf l}^{\mf a}[m]_n=\G^{\mf a}[m]_n$ and $\ov{\mf l}^{\mf a}[m]_n=\SG^{\mf a}[m]_n$.

{\bf In the remainder of the paper, we will drop the superscript $\xx$ and the subscript $\infty$ if there is no ambiguity.} For example, we write $\wcG$, $\cG[m]_n$ and $\ovcG[m]_n$ for $\wcG^\xx$, $\cG^\xx[m]_n$ and $\ovcG^\xx[m]_n$, and $\wt{\G}$, $\G[m]_n$ and $\ov\G[m]_n$ for $\wt{\G}^\xx$, $\G^\xx[m]_n$ and $\ov\G^\xx[m]_n$, respectively, where $\xx$ denotes a fixed type among $\mf{a, c,d}$. Also, we write $\G[m]$ and $\ov\G[m]$ for $\G^\xx[m]_\infty$ and $\ov\G^\xx[m]_\infty$, respectively.

Given a partition $\mu=(\mu_1,\mu_2,\ldots)$, let $\mu^\prime$ denote the conjugate partition of $\mu$. We denote by $\theta(\mu)$ the modified Frobenius
coordinates of $\mu$:
\begin{equation*}
\theta(\mu)
:=(\theta(\mu)_{1/2},\theta(\mu)_1,\theta(\mu)_{3/2},\theta(\mu)_2,\ldots),
\end{equation*}
where
$$\theta(\mu)_{i-1/2}:=\max\{\mu^\prime_i-i+1,0\}, \quad
\theta(\mu)_i:=\max\{\mu_i-i, 0\}, \quad i\in\N.
$$

Given a partition $\la=(\la_1, \la_2,\ldots)$ and $d\in\C$, we define
\begin{align}
\w{\la} &:= \sum_{r\in\hf\N}\theta(\la)_r\ep_r + d \La_0\in \wfh^{*},\label{weight:wtIm}\\
\la[m] &:=\sum_{i=1}^{m} \la^\prime_i \ep_{i-\hf}+ \sum_{j \in \N}  \left\langle \la_j-m \right\rangle  \ep_{j}
+ d\La_0\in \fh[m]^*, \label{weight:Im}\\
\ov\la[m] &:=\sum_{i=1}^{m} \la_{i}\ep_{i}+ \sum_{j \in \N} \left\langle\la^\prime_j -m \right\rangle \ep_{j-\hf}
+ d \La_0 \in \ovfh[m]^*.\label{weight:ovIm}
\end{align}
Here $\langle r  \rangle:=\max\{r, 0\}$.

Let $\w{\cP}^+(d) \subseteq \wfh^*$, $\cP[m]^+(d) \subseteq \fh[m]^*$ and $\ov{\cP}[m]^+(d) \subseteq \ovfh[m]^*$ denote the sets
of all weights of the forms \eqnref{weight:wtIm}, \eqnref{weight:Im} and \eqnref{weight:ovIm}, respectively. Let $\w{\cP}^+=\cup_{d \in \C} \w{\cP}^+(d)$, $\cP[m]^+=\cup_{d \in \C} \cP[m]^+(d)$ and $\ov{\cP}[m]^+=\cup_{d \in \C} \ov{\cP}[m]^+(d)$. By definition we have bijective maps
\begin{equation}\label{cP}
\begin{array}{l}
\w{\cP}^+ \longrightarrow \cP[m]^+\\
 \ \     \w\la\ \mapsto \ \la[m]
\end{array} \ \ \ \ \  \ \ \mbox{and}\ \ \ \ \ \
\begin{array}{l}
\w{\cP}^+ \longrightarrow \ov{\cP}[m]^+\\
 \ \     \w\la\ \mapsto \ \ovla[m]
\end{array}
\end{equation}

Recall that a partition $\la=(\la_1,\la_2,\ldots)$ is called an $(m|n)$-{\em hook partition}
if $\la_{m+1}\le n$. Let $\mathcal{P}$ and $\mathcal{P}(m|n)$ denote the set of partitions and the set of $(m|n)$-hook partitions,
respectively. Clearly, $\mathcal{P}(m|\infty)=\mathcal{P}$.
Given $d \in \C$ and $\la\in \mathcal{P}$ with $\la'\in \mathcal{P}(m|n)$ (resp., $\la\in \mathcal{P}(m|n)$), we may regard
$\la[m]  \in \fh[m]_n^*$ (resp., $\ovla[m] \in \ovfh[m]_n^*$) in a natural way.
The set of all such weights is denoted by $\cP[m]^+_n$ (resp., $\ov{\cP}[m]^+_n$).

For $\mu\in \w{\mf h}^*$, let $L(\tilde{\mf l}, \mu)$ be the irreducible highest weight $\tilde{\mf l}$-module with highest weight $\mu$.
We denote by ${\Delta}(\DG,\mu)=\mbox{Ind}^{\w\G}_{\w{\mf p}} L(\tilde{\mf l}, \mu)$ the parabolic Verma $\w\G$-module
and by $L(\DG, \mu)$ the unique irreducible quotient $\w\G$-module of $\Delta(\DG, \mu)$.
The modules $L(\mf l[m]_n, \mu)$ and $\Delta(\G[m]_n, \mu)$ (for $\mu \in \fh[m]_n^*$)
as well as $L(\ov{\mf l}[m]_n, \mu)$ and $\Delta(\SG[m]_n, \mu)$ (for $\mu \in \ovfh[m]_n^*$)
are defined analogously. 
We denote by $L(\G[m]_n, \mu)$ (resp., $L(\SG[m]_n, \mu)$) the unique irreducible quotient $\G[m]_n$(resp., $\SG[m]_n$)-module of $\Delta(\G[m]_n, \mu)$ (resp., $\Delta(\SG[m]_n, \mu)$) for $\mu \in \fh[m]_n^*$ (resp., $\ov{\fh}[m]_n^*$). 
Note that for $\xx={\mf a}$, ${\Delta}(\DG^{\mf a},\mu)=L(\DG^{\mf a},\mu)=L(\tilde{\mf l}^{\mf a}, \mu)$, ${\Delta}(\G^{\mf a}[m],\mu)=L(\G^{\mf a}[m],\mu)=L({\mf l}^{\mf a}[m], \mu)$ and ${\Delta}(\SG^{\mf a}[m],\mu)=L(\SG^{\mf a}[m],\mu)=L(\ov{\mf l}^{\mf a}[m], \mu)$.

Similar to \cite{CL10, CLW11,CLW12}, let $\w\OO$ (resp., $\OO[m]_n$ and $\ov\OO[m]_n$) be the category of
$\DG$(resp., $\G[m]_n$ and $\SG[m]_n$)-modules $M$ such that $M$ is a semisimple
$\w{\mf h}$(resp., ${\mf h}[m]_n$ and $\ov{\mf h}[m]_n$)-module with finite-dimensional weight subspaces
$M_{\gamma}$ for $\gamma \in \w{\mf h}^*$  (resp., ${\mf h}[m]_n^*$ and $\ov{\mf h}[m]_n^*$), satisfying the conditions:
   \begin{enumerate}[\normalfont(i)]
    \item $M$ decomposes over $\tilde{\mf l}$ (resp., $\mf l[m]_n$ and $\overline{\mf l}[m]_n$) as a direct sum of
$L(\tilde{\mf l}, \mu)$ (resp., $L(\mf l[m]_n, \mu)$ and $L(\ov{\mf l}[m]_n, \mu)$) for $\mu\in \w{\cP}^+$
(resp., $\cP[m]^+_n$ and $\ov{\cP}[m]^+_n$).
    \item There exist finitely many weights $\lambda^1,\lambda^2,\ldots,\lambda^k\in \w{\cP}^+$
(resp., $\cP[m]^+_n$ and $\ov{\cP}[m]^+_n$) (depending on $M$) such that if $\gamma$ is a weight of $M$, then
$\la^i-\gamma$ is a linear combination of simple roots with coefficients in $\Z_+$ for some $i$.
   \end{enumerate}
The morphisms in the categories are even homomorphisms of modules, and the categories are abelian. There is a natural $\Z_2$-gradation on each module in the categories with a compatible action of the corresponding Lie (super)algebra to be defined below.
Set
 \begin{align}\label{weight}
\w{\Xi}_n&=\sum_{r\in \w\J^+(n)} \Z_+ \ep_r  +\C\Lambda_0,\nonumber\\
 {\Xi}[m]_n&=\sum_{r\in {\J}^+_m(n)} \Z_+ \ep_r  +\C\Lambda_0,\\
 \ov{\Xi}[m]_n&=\sum_{r\in \ov{\J}^+_m(n)} \Z_+ \ep_r  +\C\Lambda_0.\nonumber
 \end{align}
For $\varepsilon=0$ or $1$ and $\Theta=\w{\Xi}$, ${\Xi}[m]_n$ or ${\ov\Xi}[m]_n$, we define
\begin{equation*}
{\Theta}(\ov \varepsilon):=
\setc[\Big]{\mu\in {\Theta} }{ \sum_{r\in \hf+\Z_+}\mu(E_{r})\equiv \varepsilon \,\,(\text{mod }2)},
\end{equation*}
where the summation is over all $r\in \hf+\Z_+$ whenever $\mu(E_{r})$ are defined.
Recall that both ${\mf l}[m]_n$ and $\ov{\mf l}[m]_n$ are isomorphic to $\gl(m|n)\oplus\C K$.
For $M\in \mathcal O[m]_n$ (resp., $\ov{\mathcal O}[m]_n$), each weight of $M$ is  a weight of a highest weight module over ${\mf l}[m]_n$  (resp., $\ov{\mf l}[m]_n$) with highest weight $\mu\in \cP[m]^+_n$ (resp., $\ov{\cP}[m]^+_n$) which is contained in $ {\Xi}[m]_n$ (resp., $\ov{\Xi}[m]_n$) (see, e.g., \cite[Proposition 3.26]{CW}). By the paragraph before Theorem 6.4 in \cite{CW}, the weights of $M$ are contained in $\w{\Xi}$ for $M\in \w{\mathcal O}$. For $M\in \w{\mathcal O}$, ${M}={M}_{\ov{0}}\bigoplus {M}_{\ov{1}}$  is a $\Z_2$-graded vector space such that
 \begin{equation}\label{wt-Z2-gradation}
{M}_{\ov{0}}:=\bigoplus_{\mu\in\w{\Xi}(\ov 0)}{M}_{\mu}\qquad\hbox{and}\qquad
{M}_{\ov{1}}:=\bigoplus_{\mu\in\w{\Xi}(\ov 1)}{M}_{\mu}.
 \end{equation}
It is clear that the $\Z_2$-gradation on $M$ is compatible with the action of $\DG$.
Similarly, we may define a $\Z_2$-gradation with a compatible action of ${\mf g}[m]_n$ (resp., $\ov{\mf g}[m]_n$) on ${M}$ for $M\in\mathcal O[m]_n$ (resp., $\ov{\mathcal O}[m]_n$).  By \cite[Theorem 3.27 and Theorem 6.4]{CW} (see also the proof of \cite[Theorem 6.2.2]{Lus}), $\w\OO$, $\OO[m]_n$ and $\ov\OO[m]_n$ are tensor categories. In particular, $\w\OO^{\mf a}$, $\OO^{\mf a}[m]_n$ and $\ov\OO^{\mf a}[m]_n$ are semisimple tensor categories. Note that the $\Z_2$-gradation on $M\otimes N$ given by \eqnref{wt-Z2-gradation} and the $\Z_2$-gradation on $M\otimes N$ induced from the $\Z_2$-gradations on $M$ and $N$ given by
\eqref{wt-Z2-gradation} are the same for $M, N\in \w\OO$ (resp., $\OO[m]_n$ and $\ov\OO[m]_n$). We summarize the results in the following proposition.

\begin{prop} Let $n\in\N\cup\{\infty\}$. Then:
\begin{enumerate} [\normalfont(i)]
\item The weights of modules in $\w\OO$ (resp., $\mathcal O[m]_n$ and $\ov{\mathcal O}[m]_n$) are contained in $\w{\Xi}$  (resp., $ {\Xi}[m]_n$ and $\ov{\Xi}[m]_n$).
\item The abelian categories $\w\OO$, $\OO[m]_n$ and $\ov\OO[m]_n$ are tensor categories.
\item $\w\OO^{\mf a}$, $\OO^{\mf a}[m]_n$ and $\ov\OO^{\mf a}[m]_n$ are semisimple tensor categories.
\end{enumerate}
\end{prop}

We also have the following proposition.

\begin{prop} \label{tensor-cat}
Let $n\in\N\cup\{\infty\}$. Then:
\begin{enumerate}[\normalfont(i)]
  \item The modules $\Delta(\DG, \la)$ and $L(\DG, \la)$ lie in $\w{\mathcal O}$ for all $\la \in \w \cP^+$.
  \item The modules $\Delta(\G[m]_n, \la)$ and $L(\G[m]_n, \la)$ lie in $\OO[m]_n$ for all $\la \in \cP[m]_n^+$.
  \item The modules $\Delta(\ov\G[m]_n, \la)$ and $L(\ov\G[m]_n, \la)$ lie in $\ov\OO[m]_n$ for all $\la \in \ov \cP[m]_n^+$.
\end{enumerate}
\end{prop}

\begin{proof}
 (i) follows from \cite[Proposition 6.7(3)]{CW}. We show (ii). For $m=0$, it follows from \cite[Proposition 6.7(1)]{CW}. For $m\not=0$, the proof is similar to that of \cite[Proposition 6.7(3)]{CW}, but here we use \cite[Theorem 3.27]{CW} (valid also for $n=\infty$) instead of \cite[Theorem 6.4]{CW}. The proof of (iii) is similar.
\end{proof}

We record here an easy but useful observation, which can be seen by the description of the weights of modules in \eqref{weight} and is essentially \cite[Lemma 2.3]{CaL}.

\begin{lem} \label{weight-decomposition}
Let $M, N\in \w\OO$ (resp., $\OO[m]_n$ and $\ov\OO[m]_n$). Suppose that $\mu$ and $\gamma$ are weights of $M$ and $N$, respectively. Then
\[
(\mu+\gamma)(E_r)=0\quad\hbox{if and only if }\quad \mu(E_r)=0\,\, \hbox{and}\,\, \gamma(E_r)=0,
 \]
for $r\in \hf \N$ (resp., ${\J}^+_m(n)$ and $\ov{\J}^+_m(n)$).
\end{lem}

Recall that we drop the subscript $n$ if $n=\infty$. For instance, we denote $\Xi[m]=\Xi[m]_\infty$ and $\OO[m]=\OO[m]_\infty$.
In view of \eqref{weight}, we immediately obtain the following (cf. \cite[Lemma 2.5]{CaL}).

\begin{lem} \label{weight-decomposition2}
Let $\mu, \gamma\in \w\Xi$. Then:
\begin{enumerate}[\normalfont(i)]
\item $\mu+\gamma \in  \Xi[m]$ if and only if $\mu \in  \Xi[m]$ and $\gamma \in  \Xi[m]$.
\item $\mu+\gamma \in  \ov{\Xi}[m]$ if and only if $\mu \in \ov{\Xi}[m]$ and $\gamma \in \ov{\Xi}[m]$.
\end{enumerate}
\end{lem}

Let $0\le k< n\leq\infty$. For $M \in \OO[m]_n$, we write $M=\bigoplus_{\mu\in {\Xi}[m]_n}M_{\mu}$. The truncation functor $\mf{tr}^n_k: \OO[m]_n \to \OO[m]_k$ is defined by
\begin{equation}\label{trun-functor}
{\mf{tr}}^n_k (M)=\bigoplus_{\nu\in {\Xi}[m]_k}M_{\nu}.
\end{equation}
For  every $f \in {\rm Hom}_{\OO[m]_n}(M, N)$, ${\mf{tr}}^n_k(f)$ is defined to be the restriction of $f$ to ${\mf{tr}}^n_k(M)$.
The truncation functor $\ov{\mf{tr}}^n_k: \ov\OO[m]_n \to \ov\OO[m]_k$ can be defined in a similar way. It is clear that ${\mf{tr}}^n_k$ and $\ov{\mf{tr}}^n_k$ are exact functors. By \lemref{weight-decomposition}, we immediately have the following lemma.

\begin{lem} \label{trun-tensor}
For $0\le k<n \leq\infty$, ${\mf{tr}}^n_k$ and $\ov{\mf{tr}}^n_k$ are tensor functors.
\end{lem}

The following proposition will be useful later on. It may be proved in a similar way to the proof of \cite[Lemma 3.2]{CLW11}.
\begin{prop} \label{trun}
Let $0\le k<n \leq\infty$ and ${\mu}\in {\cP}[m]_n^+$. Suppose $V_i=\Delta(\G[m]_i,\mu)$ or $L(\G[m]_i,\mu)$ for $i=k, n$. Then
 $$
 {\mf{tr}}^n_k (V_n)=
 \begin{cases}
       V_k & \mbox{if}\ \ {\mu}\in  {\cP}[m]_k^+; \\
       0 & \mbox{otherwise}.
     \end{cases}
 $$
A similar statement holds for $\ov{\mf{tr}}^n_k$.
\end{prop}

Given ${M}=\bigoplus_{\gamma\in\wfh^*}{M}_\gamma\in\w\OO$,  we define
\begin{align*}
T_{[m]}({M})= \bigoplus_{\gamma\in{{\fh[m]^*}}}{M}_\gamma \qquad \hbox{and}\qquad
 \ov{T}_{[m]}({M})=\bigoplus_{\gamma\in{\ovfh[m]^*}}{M}_\gamma.
\end{align*}

For $M, N \in \w\OO$ and $f\in {\rm Hom}_{\w\OO}(M,N)$, $T_{[m]}(f)$ and $\ov{T}_{[m]}({f})$ are defined to be the restrictions of $f$ to $T_{[m]}({M})$ and $\ov{T}_{[m]}({M})$, respectively. Note that $T_{[m]}({f}): T_{[m]}({M}) \to T_{[m]}({N})$ and $\ov{T}_{[m]}({f}):   \ov{T}_{[m]}({M}) \to  \ov{T}_{[m]}({N})$ are respectively a $\G[m]$-homomorphism and a $\SG[m]$-homomorphism. Moreover, the functors $T_{[m]}: \w\OO\rightarrow\OO[m]$ and $\ov{T}_{[m]}: \w\OO\rightarrow\ov\OO[m]$ are exact (cf. \cite[Proposition 6.15]{CW}).

By \lemref{weight-decomposition2}, we see that $T_{[m]}(M\otimes N)=T_{[m]}(M) \otimes T_{[m]}(N)$ and $\ov{T}_{[m]}(M\otimes N)= \ov{T}_{[m]}(M)\otimes \ov{T}_{[m]}(N)$ for all $M, N\in \w{\mathcal O}$, and so $T_{[m]}$ and $\ov{T}_{[m]}$ are tensor functors. We have the following result.
\begin{thm} \label{SD}
The following statements hold:
\begin{enumerate}[\normalfont(i)]
\item For each $m\in\Z_+$, $T_{[m]}: \w\OO\rightarrow\OO[m]$ is an equivalence of tensor categories.
\item For each $m\in\Z_+$, $\ov{T}_{[m]}: \w\OO\rightarrow\ov\OO[m]$ is an equivalence of tensor categories.
\end{enumerate}
Moreover, $T_{[m]}$ and $\ov{T}_{[m]}$ send parabolic Verma modules to parabolic Verma modules and irreducible modules to irreducible modules. More precisely, for $\wla \in \w{\cP}^+$, we have
\begin{eqnarray*}
&& T_{[m]}\big{(} {\Delta}(\DG,\wla)  {)}={\Delta}(\G[m], \la[m]), \quad T_{[m]}\big{(} L(\DG,\wla) \big{)}=L(\G[m], \la[m]), \\
&&  \ov{T}_{[m]}\big{(} {\Delta}(\DG,\wla) {)}={\Delta}(\SG[m],\ov{\la}[m]), \quad \ov{T}_{[m]} \big{(} L(\DG,\wla) \big{)}=L(\SG[m], \ov{\la}[m]).
\end{eqnarray*}
\end{thm}

Theorem \ref{SD} can be proved along the lines of the proof of the super duality in \cite{CL10, CLW11} by making use of the techniques in \cite{CLW12} (see particularly \cite[Section 7]{CLW12}). The proof is omitted here. We also call \thmref{SD} super duality.
We only need that $T_{[m]}$ and $\ov{T}_{[m]}$ are tensor functors in this paper. Note that $T_{[0]}$ and $\ov{T}_{[0]}$ are the functors $T$ and $\ov T$ in the degenerate case defined in \cite{CL10, CLW11}. \thmref{SD} also implies that the tensor categories $\OO[0]$ and $\ov\OO[0]$ for Lie algebras and the tensor categories $\OO[m]$ and $\ov\OO[m]$ for Lie superalgebras are equivalent for $m \in \N$.

\section{Unitarizable $\ovcG^\xx[m]_n$-modules} \label{unitarizable}
The notion of unitarizable modules will play an important role in our study of (super) Gaudin Hamiltonians. In this section, we start by introducing $*$-structures on $\ovcG^\xx[m]_n$ and $\SG^\xx[m]_n$. We then describe the unitarizable $\ovcG^\xx[m]_n$-modules to be studied in this paper.

We first recall some basic facts about $*$-superalgebras and their unitarizable modules. A \emph{$*$-superalgebra} is an associative superalgebra $A$ together with an anti-linear anti-involution $\omega: A \to A$ of degree $\ov 0$. A homomorphism $f: (A,\omega) \to (A^\prime, \omega^\prime)$ of $*$-superalgebras is a homomorphism of superalgebras satisfying $\omega^\prime \circ f = f \circ \omega$.
Let $(A,\omega)$ be a $*$-superalgebra, and let $V$ be a $\Z_2$-graded $A$-module.  A Hermitian form $\langle\cdot|\cdot\rangle$ on $V$ is said to be \emph{contravariant} if $\langle av | v'\rangle=\langle v |\omega(a)v'\rangle$, for all $a\in A$, $v,v'\in V$.
An $A$-module equipped with a positive definite contravariant Hermitian form is called a \emph{unitarizable} $A$-module.

A Lie superalgebra $\G$ is said to admit a \emph{$*$-structure} if $\G$ is equipped with an anti-linear anti-involution $\omega$ of degree $\ov 0$. In this case, $\omega$ is also called a $*$-structure on $\G$. A homomorphism $f: (\G, \omega) \to (\G^\prime, \omega^\prime)$ of Lie superalgebras with $*$-structures is a homomorphism of Lie superalgebras satisfying $\omega^\prime \circ f = f \circ \omega$.
Moreover, it is clear that $\omega$ is a $*$-structure on $\G$ if and only if the natural extension of $\omega$ to the universal enveloping algebra $U(\G)$ of $\G$ is an anti-linear anti-involution.
Let $(\G, \omega)$ be a Lie superalgebra with $*$-structure, and let $V$ be a $\Z_2$-graded $\G$-module. A Hermitian form $\langle\cdot|\cdot\rangle$ on $V$ is said to be \emph{contravariant} if $\langle xv|v'\rangle=\langle v|\omega(x)v'\rangle$, for all $x\in \G$, $v,v'\in V$.
A $\G$-module equipped with a positive definite contravariant Hermitian form is called a \emph{unitarizable} $\G$-module.
Note that a $\G$-module $V$ is a unitarizable $\G$-module if and only if $V$ is a unitarizable $U(\G)$-module.

 \subsection{$*$-structures on $\ovcG^\xx[m]_n$ and $\SG^\xx[m]_n$} \label{star}
 Recall that the Lie superalgebra $\widehat\gl(\w{V})$ is the central extension  of ${\gl}(\w{V})$ with a basis $\{E_{i,j}, K\,|\, i,j\in \hf \Z\}$. It admits a $*$-structure $\omega$ defined by (cf. \cite[p. 421]{LZ06})
$$\sum_{i, j \in \hf \Z} a_{ij} E_{i, j} \mapsto \sum_{i, j \in \hf \Z} (-1)^{[i]+[j]} \ov{a}_{ij} E_{j, i} \qquad\hbox{and}\qquad K \mapsto K.$$
Here $\ov{a}_{ij}$ denotes the complex conjugate of $a_{ij} \in \C$ and
 $$
[i]:=
 \begin{cases}
       1 & \mbox{if}  \ \  -i \in \hf+\Z_+;\\
       0 & \mbox{if}  \ \  -i \in\hf\Z\backslash (\hf+\Z_+).
     \end{cases}
 $$
It is evident from the spanning sets (i.e., the sets of elements described in Sections \ref{type-a}, \ref{type-c} and \ref{type-d} together with $K$) of the Lie superalgebras $\DG^{\mf{\xx}}$, $\G^{\mf{\xx}}[m]_n$ and $\SG^{\mf{\xx}}[m]_n$ that the restrictions of $\omega$ to $\DG^{\mf{\xx}}$, $\G^{\mf{\xx}}[m]_n$ and $\SG^{\mf{\xx}}[m]_n$, denoted also by $\omega$, give $*$-structures on these Lie superalgebras.

Since $\omega$ is a $*$-structure on $\widehat\gl(\w V^\times)$ and $\widehat{\varphi}$ is an involution on $\widehat\gl(\w V^\times)$, the map $\omega':=\widehat{\varphi}\circ\omega\circ\widehat{\varphi}$ is a $*$-structure on $\widehat\gl(\w V^\times)$. More precisely,
$$
\omega'(E_{r,s})=(-1)^{\tau_r+\tau_s}E_{s,r},\quad \hbox{for $r,s\in \hf\Z^*$},\qquad\hbox{and}\qquad \omega'(K)=K,
$$
where
$$
\tau_r:=  \begin{cases}
       1 & \mbox{if}  \ \  -r \in \N;\\
       0 &  \mbox{if}  \ \ -r \in\hf\Z^*\backslash \N.
     \end{cases}
$$
Via the isomorphism $\widehat{\varphi}:  \ov\G^{\mf{c}}[m]_n\longrightarrow \G^{\mf{d}}[m]_n$ given in \lemref{lem:phi_ext}, an anti-linear anti-involution $\omega$ on $\G^{\mf{d}}[m]_n$ pulls back to an anti-linear anti-involution $\omega':=\widehat{\varphi}^{-1}\circ\omega\circ\widehat{\varphi}$ on $\ov\G^{\mf{c}}[m]_n$ while, via the isomorphism $\widehat{\varphi}^{-1}:  \ov\G^{\mf{d}}[m]_n\longrightarrow\G^{\mf{c}}[m]_n$, an anti-linear anti-involution $\omega$ on $\G^{\mf{c}}[m]_n$ pulls back to an anti-linear anti-involution $\omega':=\widehat{\varphi}\circ\omega\circ\widehat{\varphi}^{-1}$ on $\ov\G^{\mf{d}}[m]_n$. In other words, the map $\widehat{\varphi}$  (resp., $\widehat{\varphi}^{-1}$) gives an isomorphism of Lie superalgebras with $*$-structures from $(\ov\G^{\mf{c}}[m]_n, \omega')$ (resp., ($\ov\G^{\mf{d}}[m]_n, \omega')$) to $(\G^{\mf{d}}[m]_n, \omega)$ (resp., $(\G^{\mf{c}}[m]_n, \omega)$). Note that the $*$-structure $\omega'$ on $\ov\G^{\mf{c}}[m]_n$ (resp., $\ov\G^{\mf{d}}[m]_n$) is the restriction of $\omega'$ defined on $\widehat\gl(\w V^\times)$.

 It is clear that the $*$-structure $\omega$ (resp., $\omega'$) induces  a $*$-structure, denoted also by $\omega$ (resp., $\omega'$),  on $\gl(\w V^\times)$. For $\xx=\mf{ c, d}$, the  restriction of $\omega$ (resp., $\omega'$) to $\ovcG^{\xx}[m]_n$ gives  a $*$-structure on $\ovcG^{\xx}[m]_n$, denoted also by $\omega$ (resp., $\omega'$). We have the following proposition.

 \begin{prop}\label{*-homo_cd}
  For $\xx=\mf {c, d}$, the restriction of the isomorphism $\iota: \ovcG^{\xx}[m]_n \oplus \C K \longrightarrow \SG^{\xx}[m]_n$ defined by \eqref{iso-e} to $\ovcG^{\xx}[m]_n$ gives two monomorphisms $\ov\iota: (\ovcG^{\xx}[m]_n, \omega) \longrightarrow (\SG^{\xx}[m]_n, \omega)$ and $\ov\iota': (\ovcG^{\xx}[m]_n, \omega') \longrightarrow (\SG^{\xx}[m]_n,  \omega')$ of Lie superalgebras with $*$-structures.
 \end{prop}

From now on, we denote by $L(\cG^\xx[m]_n, \mu)$ (resp., $L(\ovcG^\xx[m]_n, \mu)$) the irreducible highest weight $\cG^\xx[m]_n$(resp., $\ovcG^\xx[m]_n$)-module with highest weight $\mu \in {\bf h}^\xx[m]_n^*$ (resp., $\ov{\bf h}^\xx[m]_n^*$). 
Recall the irreducible module $L(\DG^\xx, \mu)$ (resp., $L(\G^\xx[m]_n, \mu)$ and $L(\SG^\xx[m]_n, \mu)$) for $\mu \in \wfh^{\xx *}$  (resp., $\fh^\xx [m]_n^{*}$ and $\ov{\fh}^\xx [m]_n^{*}$) defined in Section \ref{category-O}.

\subsection{Unitarizable $\ovcG^{\mf a}[m]_n$-modules} \label{unitarizable-a}
The $2$-cocycle in \eqref{central extension} is zero when it restricts to $\ovcG^{\mf a}[m]_n \times \ovcG^{\mf a}[m]_n$, and we may identify $\ovcG^{\mf a}[m]_n$ as a subalgebra of $\SG^{\mf a}[m]_n$. Thus, the restriction of $\omega$ on $\SG^{\mf a}[m]_n$ to $\ovcG^{\mf a}[m]_n$ is a $*$-structure on $\ovcG^{\mf a}[m]_n$, which we also denote by $\omega$. More precisely,
$$\omega (E_{i, j} )=E_{j, i}  \qquad \hbox{for} \quad i,j\in \ov{\J}^+_m(n).$$
Recall that $\mathcal{P}(m|n)$ denotes the set of $(m|n)$-hook partitions.
Define 
\begin{align*}
 {\mc Q}^{\mf a, I}(m|n)&:=\setc*{\sum_{i=1}^{m} \la^\prime_i \ep_{i-\hf}+\sum_{j=1}^n \left\langle \la_j-m \right\rangle  \ep_{j} \in {{\bf h}}^{\mf a}[m]_n^{*}}{ \la' \in \cP(m|n)}, \\
 \ov{\mc Q}^{\mf a, I}(m|n)&:=\setc*{\sum_{i=1}^m \la_i \ep_i+\sum_{j=1}^n \langle \la^\prime_j -m \rangle\ep_{j-\hf}\in {\ov{\bf h}}^{\mf a}[m]_n^{*}}{ \la \in \cP(m|n)}.
\end{align*}

The following proposition is well known (see, e.g., \cite[Theorems 3.2 and 3.3]{CLZ} with $p=q=0$).

\begin{prop}
For $\xi \in \ov{\mc Q}^{\mf a, I}(m|n)$, $L(\ovcG^{\mf a}[m]_n,\xi)$ is a unitarizable $\ovcG^{\mf a}[m]_n$-module with respect to the $*$-structure $\omega$.
\end{prop}

\begin{rem}
The modules appearing in the above proposition are exactly the irreducible polynomial modules over $\ovcG^{\mf a}[m]_n$ (see, e.g., \cite[Proposition 3.26]{CW}).
\end{rem}

Recall that $\la[m]$ and $\ov\la[m]$ are defined in \eqref{weight:Im} and \eqref{weight:ovIm}, respectively.
For $n\in\N\cup\{\infty\}$, let
\begin{align*}
 Q^{\mf a}(m|n)&:=\{\la[m]\in {\mf h}^{\mf a}[m]_n^{*}\,|\,  \la'\in \cP(m|n), \, d=0\}, \\
 \ov Q^{\mf a}(m|n)&:=\{\ov\la[m]\in \ov{\mf h}^{\mf a}[m]_n^{*}\,|\,  \la\in \cP(m|n), \, d=0\}.
\end{align*}
The conditions for $\la, \la'$ in $\cP(m|n)$ are unnecessary when $n=\infty$.
The restrictions of the bijections given in \eqnref{cP} give  the following bijections:
\begin{equation}\label{Qa}
   Q^{\mf a}(m'|\infty)\leftrightarrow  \wt Q^{\mf a}:=\w{\cP}^+(0) \leftrightarrow \ov Q^{\mf a}(m|\infty)  \qquad \qquad \hbox{for $m, m'\in \Z_+$.}
\end{equation}
These sets will be used in Section \ref{diag}.

\subsection{Unitarizable modules over $\ovcG^{\mf c}[m]_n$ and $\ovcG^{\mf d}[m]_n$} \label{unitarizable-cd}
In this subsection, we will restrict our attention to $\xx=\mf {c, d}$. There are two types of unitarizable highest weight modules over $\ovcG^{\xx}[m]_n$ corresponding to the $*$-structures $\omega$ and $\omega'$ defined above.

The Lie superalgebra ${\mc C}^f$ (resp., ${\mc D}^f$) defined in \cite{LZ06} is our $\wcG^{\mf c}$ (resp., $\wcG^{\mf d}$) while $\widehat{\mc C}^f$ (resp., $\widehat{\mc D}^f$) is our $\DG^{\mf c}$ (resp., $\DG^{\mf d}$). Also, the set of the unitarizable quasi-finite irreducible highest weight modules over $\widehat{\mc C}$ (resp., $\widehat{\mc D}$) described in [Proposition 5.8]\cite{LZ06}(resp., [Proposition 5.9]\cite{LZ06}) are the set of unitarizable irreducible highest weight modules over $\widehat{\mc C}^f$ (resp., $\widehat{\mc D}^f$). Recall that $\w\la$ is defined in \eqref{weight:wtIm}.
Let
\begin{align*}
\w Q^{\mf c}&:=\setc*{\w\la\in \w{\mf h}^{\mf c\, *}}{ \la_1\le d,\,\, \la\in \cP, \, d\in \Z_+ }, \\
\w Q^{\mf d}&:=\setc[\Big]{\w\la\in \w{\mf h}^{\mf d\, *}}{ \la_1+ \la_2\le 2d,\,\, \la\in \cP, \, d\in \hf \Z_+ }.
\end{align*}
Reformulating the results in \cite{LZ06} in terms of our notations, we obtain the following proposition.

\begin{prop}\label{prop:wUweight}
  \begin{enumerate} [\normalfont(i)]
  \item An irreducible highest weight $\DG^{\mf c}$-module $M$ is unitarizable with respect to $\omega$ if and only if $M\cong L(\DG^{\mf c}, \xi)$ for some $\xi\in \w Q^{\mf c}$.
     \item An irreducible highest weight $\DG^{\mf d}$-module $M$ is unitarizable with respect to $\omega$ if and only if $M\cong L(\DG^{\mf d}, \xi)$ for some $\xi\in \w Q^{\mf d}$.
    \end{enumerate}
\end{prop}
Recall $\la[m]$ and $\ov\la[m]$ defined in \eqref{weight:Im} and \eqref{weight:ovIm}, respectively.
For $n\in\N\cup\{\infty\}$, let
\begin{align*}
 Q^{\mf c}(m|n)&:=\setc*{\la[m]\in {\mf h}^{\mf c}[m]_n^{*} }{ \la_1\le d,\,\, \la'\in \cP(m|n), \, d\in \Z_+ }, \\
 Q^{\mf d}(m|n)&:=\setc[\Big]{\la[m]\in {\mf h}^{\mf d}[m]_n^{*} }{  \la_1+ \la_2\le 2d,\,\, \la'\in \cP(m|n), \, d\in \hf \Z_+ },\\
\ov Q^{\mf c}(m|n)&:=\setc[\Big]{\ov\la[m]\in \ov{\mf h}^{\mf c}[m]_n^{*} }{  \la_1\le d,\,\, \la\in \cP(m|n), \, d\in \Z_+ }, \\
\ov Q^{\mf d}(m|n)&:=\setc[\Big]{\ov\la[m]\in \ov{\mf h}^{\mf d}[m]_n^{*} }{  \la_1+ \la_2\le 2d,\,\, \la\in \cP(m|n), \, d\in \hf \Z_+ }.
\end{align*}
The conditions for $\la, \la'$ in $\cP(m|n)$ are unnecessary when $n=\infty$.
The restrictions of the bijections given in \eqnref{cP} give  the following bijections:
\begin{equation}\label{Qx}
   Q^{\xx}(m'|\infty)\leftrightarrow  \wt Q^{\xx} \leftrightarrow \ov Q^{\xx}(m|\infty) \qquad \qquad \hbox{for $\xx={\mf{c, d}}$ and $m, m'\in \Z_+$.}
\end{equation}
The proof of the following is straightforward.
\begin{lem}\label{subU}
Let $\G$ be a Lie superalgebra with $*$-structure $\sigma$. Assume that $\mf u$ is a subalgebra of $\G$ such that the restriction $\sigma|_{\mf u}$ of $\sigma$ to $\mf u$ is a $*$-structure on $\mf u$. Let $V$ be a unitarizable $\G$-module with respect to $\sigma$. If $W$ is a $\mf u$-submodule of $V$, then $W$ is a unitarizable $\mf u$-module with respect to $\sigma|_{\mf u}$.
\end{lem}
The following proposition is a direct consequence of \propref{prop:wUweight} and \lemref{subU}.
\begin{prop}\label{prop:ovUweight}
  \begin{enumerate} [\normalfont(i)]
  \item For $\xi\in  Q^{\mf c}(m|n)$, $L(\G^{\mf c}[m]_n, \xi)$ is a unitarizable $\G^{\mf c}[m]_n$-module with respect to $\omega$.
  \item For $\xi\in  Q^{\mf d}(m|n)$, $L(\G^{\mf d}[m]_n, \xi)$ is a unitarizable $\G^{\mf d}[m]_n$-module with respect to $\omega$.
    \item For $\xi\in \ov Q^{\mf c}(m|n)$, $L(\ov\G^{\mf c}[m]_n, \xi)$ is a unitarizable $\SG^{\mf c}[m]_n$-module with respect to $\omega$.
  \item For $\xi\in \ov Q^{\mf d}(m|n)$, $L(\ov\G^{\mf d}[m]_n, \xi)$ is a unitarizable $\SG^{\mf d}[m]_n$-module with respect to $\omega$.
    \end{enumerate}
\end{prop}

\begin{proof}
  To show (iii), let $\la\in \cP(m|n)$ and $d\in \Z_+$ be such that $\la_1\le d$. Then $\ov\la[m]\in \ov Q^{\mf c}(m|n)$ and
 $  L(\SG^{\mf c}[m]_n, \ov\la[m])=\ov{\mf {tr}}^\infty_n \big(\ov T_{[m]}( L(\w\G^{\mf c}, \w\la)) \big)$.
  By \lemref{subU}, $  L(\SG^{\mf c}[m]_n, \ov\la[m])$ is a unitarizable module with respect to $\omega$. The other parts can be proved by a similar argument.
\end{proof}
\begin{defn}
For $\xx=\mf{c, d}$, a $\ovcG^{\xx}[m]_n$-module $M$ is said to be a {\em unitarizable module of type I (resp., II)} if $M$ is unitarizable with respect to the $*$-structure $\omega$ (resp., $\omega'$).
\end{defn}
Let
$$
\mathbbm{1}_{m|n}=\sum_{i=1}^m \ep_i-\sum_{j=1}^n \ep_{j-\hf}.
$$
For $\la \in \mathcal{P}(m|n)$, we define
$$
\ovla=\sum_{i=1}^m \la_i \ep_i+\sum_{j=1}^n \langle \la^\prime_j -m \rangle\ep_{j-\hf}.
$$
Let
\begin{align*}
\ov{\mc Q}^{\mf c, I}(m|n)&:=\setc[\Big]{\ov\la-d\mathbbm{1}_{m|n}\in \ov{\bf h}^{\mf c}[m]_n^{*} }{  \la_1\le d,\,\, \la\in \cP(m|n), \, d\in \Z_+ }, \\
\ov{\mc Q}^{\mf d, I}(m|n)&:=\setc[\Big]{\ov\la-d\mathbbm{1}_{m|n}\in \ov{\bf h}^{\mf d}[m]_n^{*} }{  \la_1+ \la_2\le 2d,\,\, \la\in \cP(m|n), \, d\in \hf \Z_+ },\\
\ov{\mc Q}^{\mf c, II}(m|n)&:=\setc[\Big]{\ov\la+d\mathbbm{1}_{m|n}\in \ov{\bf h}^{\mf c}[m]_n^{*} }{  \la'_1+ \la'_2\le 2d,\,\, \la \in \cP(m|n), \, d\in \hf \Z_+ },\\
 \ov{\mc Q}^{\mf d, II}(m|n)&:=\setc[\Big]{\ov\la+d\mathbbm{1}_{m|n}\in \ov{\bf h}^{\mf d}[m]_n^{*} }{  \la'_1\le d,\,\, \la \in \cP(m|n), \, d\in \Z_+ }.
\end{align*}

\begin{prop}\label{prop: U_typeI_II}
  \begin{enumerate} [\normalfont(i)]
  \item For $\xi\in  \ov{\mc Q}^{\mf c, I}(m|n)$, $L(\ovcG^{\mf c}[m]_n, \xi)$ is a unitarizable $\ovcG^{\mf c}[m]_n$-module of type I.
  \item For $\xi\in  \ov{\mc Q}^{\mf c, II}(m|n)$, $L(\ovcG^{\mf c}[m]_n, \xi)$ is a unitarizable $\ovcG^{\mf c}[m]_n$-module of type II.
   \item For $\xi\in  \ov{\mc Q}^{\mf d, I}(m|n)$, $L(\ovcG^{\mf d}[m]_n, \xi)$ is a unitarizable $\ovcG^{\mf d}[m]_n$-module of type I.
  \item For $\xi\in  \ov{\mc Q}^{\mf d, II}(m|n)$, $L(\ovcG^{\mf d}[m]_n, \xi)$ is a unitarizable $\ovcG^{\mf d}[m]_n$-module of type II.
    \end{enumerate}
\end{prop}

\begin{proof}
  To show (i), let $\la\in \cP(m|n)$ and $d\in \Z_+$ be such that $\la_1\le d$. The $\ovcG^{\mf c}[m]_n$-module structure, induced by the monomorphism $\ov\iota: (\ovcG^{\mf c}[m]_n, \omega) \longrightarrow (\SG^{\mf c}[m]_n, \omega)$ of Lie superalgebras with $*$-structures given in \propref{*-homo_cd}, on $L(\SG^{\mf c}[m]_n, \ov\la[m])$ is isomorphic to $L(\ovcG^{\mf c}[m]_n, \ov\la-d\mathbbm{1}_{m|n})$. By \propref{prop:ovUweight}(iii), part (i) follows.
  To show (ii), let $\la \in \cP(m|n)$ and $d\in \hf\Z_+$ be such that $\la'_1+ \la'_2\le 2d$. The map $\ov\iota': (\ovcG^{\mf c}[m]_n , \omega') \longrightarrow (\SG^{\mf c}[m]_n,  \omega')$ given in \propref{*-homo_cd} is a monomorphism of Lie superalgebras with $*$-structures, and we have shown that the map $\widehat{\varphi}$ given in \lemref{lem:phi_ext} is an isomorphism of Lie superalgebras with $*$-structures from $(\ov\G^{\mf{c}}[m]_n, \omega')$ to $(\G^{\mf{d}}[m]_n, \omega)$. The $\ovcG^{\mf c}[m]_n$-module structure induced by $\widehat{\varphi}\circ\ov\iota' $ on $L(\G^{\mf  d}[m]_n, \la'[m])$ is isomorphic to $L(\ovcG^{\mf c}[m]_n, \ov\la+d\mathbbm{1}_{m|n})$. By \propref{prop:ovUweight}(ii), part (ii) follows. The other parts can be proved by a similar argument.
\end{proof}

\begin{rem}\label{inf_U}
  For $n\in \N$, $\ovcG^{\mf c}[0]_n$ and $\ovcG^{\mf d}[0]_n$ are Lie algebras, and the highest weights in (ii) (resp., (iv)) of the above proposition are exactly the highest weights appearing in the classification of infinite-dimensional unitarizable highest weight modules over $\ovcG^{\mf c}[0]_n$ (resp., $\ovcG^{\mf d}[0]_n$) with integral (resp., half integral and integral) values  given in \cite[Sections 8 and 9]{EHW} (see also \cite[Theorem 2.5]{HLT}).
\end{rem}

\section{Gaudin Hamiltonians on modules over $\w\G$, $\G[m]_n$ and $\ov\G[m]_n$} \label{GH}
In this section, we define the Casimir symmetric tensors for the Lie (super)algebras $\w\G$, $\G[m]_n$ and $\SG[m]_n$ of infinite and finite ranks, and introduce the (super) Gaudin Hamiltonians associated to these Lie (super)algebras. Our main goal is to show that the set of singular eigenvectors of each Gaudin Hamiltonian for $\DG$ is in one-to-one correspondence with the set of singular eigenvectors of the corresponding Gaudin Hamiltonian for $\G[m]$ (resp., $\SG[m]$). Besides, each eigenvector and its corresponding eigenvector, under the one-to-one correspondence, have the same eigenvalue. We also show that the singular eigenvectors of the Gaudin Hamiltonians for $\G[m]$ (resp., $\SG[m]$) and those of $\G[m]_n$ (resp., $\SG[m]_n$), for $n \in \N$, are related by truncation functors.

Recall that $\w{\Phi}^+_n$ (resp., ${\Phi[m]}^+_n$ and $\ov{\Phi}[m]^+_n$) denote the set of positive roots of $\DG_n$ (resp., $\G[m]_n$ and $\SG[m]_n$).
First of all, we have the following lemma. It is analogous to  \cite[Lemma 3.1]{CaL} and can be proved similarly.

 \begin{lem} \label{finitesum}
 Let $0\le k<n \leq\infty$. If $v$ is a weight vector of weight $\mu$ in $M \in\w\OO$ (resp., $\OO[m]_n$ and $\ov\OO[m]_n$) such that $\mu\in \w\Xi_k$ (resp., $\Xi[m]_k$ and $\ov\Xi[m]_k$), then $E_\beta v= 0$ for all $\beta\in \w{\Phi}^{+}\backslash\w{\Phi}_k^{+}$ (resp., ${\Phi}[m]_n^{+}\backslash{\Phi}[m]_k^{+}$ and $\ov{\Phi}[m]_n^{+}\backslash\ov{\Phi}[m]_k^{+}$) and $E_i v= 0$ for $i>k$.
In particular, for each $v\in M$ and $M \in \w{\mathcal {O}}$, ${\mathcal {O}}[m]$ or $\ov{\mathcal {O}}[m]$, there are only finitely many $E_\beta$ and $E_i$ such that $E_\beta v\not= 0$ and $E_i v\not= 0$.
 \end{lem}

Let $(\cdot, \cdot)$ denote the bilinear form  on $\gl(\w{V})$ defined by
 \begin{equation*}\label{str}
 ( A,B):=  \mbox{Str}(AB)\quad \hbox{for $A,B\in \gl(\w{V}).$}
\end{equation*}
It is non-degenerate invariant even supersymmetric. The restriction of the above bilinear form to ${\wcG}$ (resp., $\cG[m]_n$ and $\ovcG[m]_n$) is also a non-degenerate invariant even supersymmetric bilinear form. We denote
 $$
 \hbox{$\langle\cdot, \cdot\rangle:=(\cdot, \cdot)\,\,$ on ${\wcG}^{\mf{a}}$, $\cG^{\mf{a}}[m]_n$ and $\ovcG^{\mf{a}}[m]_n$}, \quad \hbox{and} \quad  \hbox{$\langle\cdot, \cdot\rangle:=\hf(\cdot, \cdot)\,\,\,$ on the other cases.}
 $$

It is clear that $\langle E_i,E_i\rangle=(-1)^{2i}$ for any $i\in \hf \N$. For each root $\beta \in \w{\Phi}^+$ (resp., ${\Phi[m]}^+_n$ and $\ov{\Phi}[m]^+_n$), we choose root vectors $E_\beta$ and $E^\beta$ of weights $\beta$ and $-\beta$, respectively, such that
 \begin{equation*}
   \langle E_{\beta},E^{\beta}\rangle=1.
 \end{equation*}
Also, $\langle E^{\beta}, E_{\beta} \rangle=(-1)^{| E_\beta |}$, where $|E_{\beta}|$ is the parity of $E_{\beta}$.

By identifying $\w{\mathcal {G}}$ (resp., $\mathcal {G}_n$ and $\ov{\mathcal {G}}_n$) as a subspace of $\DG$ (resp., $\G_n$ and $\SG_n$), the Casimir symmetric tensors for $\DG$, $\G[m]_n$ and $\SG[m]_n$ are defined by (cf. \cite[Section 3.1]{CaL})
  \begin{align*}
 {\w{\Omega}}:=&\sum_{\beta\in\w{\Phi}^+}(E^\beta\otimes E_{\beta}+(-1)^{|E_\beta|} E_\beta \otimes E^{\beta})  \\
 &+\sum_{j \in \hf \N} \left( (-1)^{2j} E_j\otimes E_j-(K\otimes E_j+E_j\otimes K) \right),\\
{\Omega}[m]_n:=&\sum_{\beta \in  {\Phi[m]}^+_n} \big(E^\beta\otimes E_{\beta}+(-1)^{|E_\beta|} E_\beta\otimes E^{\beta}) \\
 &+\sum_{j\in{\J}^+_m(n)} \left( (-1)^{2j} E_j\otimes E_j-(K\otimes E_j+E_j\otimes K) \right),\\
 {\ov{\Omega}}[m]_n:=&\sum_{\beta \in \ov{\Phi}[m]^+_n}(E^\beta\otimes E_{\beta}+(-1)^{|E_\beta|} E_\beta \otimes E^{\beta})  \\
 &+\sum_{j\in\ov{\J}^+_m(n)} \left( (-1)^{2j} E_j\otimes E_j-(K\otimes E_j+E_j\otimes K) \right).
  \end{align*}

\begin{rem} \label{O[m]}
For $\beta \in {\Phi[m]}^+_n$, either $\beta \in \w{\Phi}^+$ or $-\beta \in \w{\Phi}^+$. Assume that $-\beta \in \w{\Phi}^+$. Then $\beta$ is an odd root. We readily see that
$E_{-\beta}=aE^{\beta}$ and $E^{-\beta}=-a^{-1} E_{\beta}$ for some nonzero scalar $a$. It follows that $E^{-\beta} \otimes E_{-\beta}- E_{-\beta}\otimes E^{-\beta}=E^\beta\otimes E_{\beta}- E_\beta\otimes E^{\beta}$. In other words, ${\Omega}[m]_n$ is a partial sum of ${\w{\Omega}}$. Similarly, $\ov{\Omega}[m]_n$ is a partial sum of ${\w{\Omega}}$ as well.
\end{rem}

 By \lemref{finitesum}, the Casimir symmetric tensors $\w{\Omega}$, ${\Omega}[m]_n$ and $\ov{\Omega}[m]_n$ are well defined operators on $M\otimes N$, for $M, N\in\w\OO$, $\OO[m]_n$ and $\ov\OO[m]_n$, respectively.

Fix $\ell \in \N$ with $\ell \geq 2$. For $M_1, \ldots, M_\ell \in \w\OO$, let
$$
M:=M_1\otimes\cdots\otimes  M_\ell.
$$
Let us introduce some more notation.
For $x \in \DG$  (resp., $\G[m]_n$ and $\ov{\G}[m]_n$)  and $i=1, \ldots, \ell$, let
$$
x^{(i)}=\underbrace{1\otimes\cdots\otimes1\otimes \stackrel{i}{x}\otimes1\otimes\cdots\otimes1}_{\ell}.
$$
For any operator $A=\sum_{r\in I} x_r\otimes y_r$, where $x_r, y_r \in \w{\G}$  (resp., $\G[m]_n$ and $\ov{\G}[m]_n$), and for any distinct $i,j \in \{ 1,\ldots, \ell \}$, we define
\begin{equation}\label{(ij)}
A^{(ij)}=\sum_{r\in I}x_r^{(i)}y_r^{(j)}.
\end{equation}

For any $i=1, \ldots, \ell$ and any distinct complex numbers $z_1, \ldots, z_\ell$, the (quadratic) Gaudin Hamiltonian $\wcH^i$ is defined by
\begin{equation}\label{Hamiltonian-C}
\wcH^i=\sum_{\substack{j=1 \\ j\not=i}}^\ell \frac{\w{\Omega}^{(ij)}}{z_i-z_j}.
\end{equation}
Note that it is a linear endomorphism on $M$. The Gaudin Hamiltonians $\cH^i[m]_n$ and $\ovcH^i[m]_n$ are defined by replacing $\w{\Omega}$ with $\Omega[m]_n$ and $\ov{\Omega}[m]_n$, respectively, and they are well-defined linear endomorphisms on $M_1\otimes\cdots\otimes  M_\ell$ for $M_1, \ldots, M_\ell$ in $\OO[m]_n$ and in $\ov\OO[m]_n$, respectively.

For any $N \in \w\OO$  (resp., $\OO[m]_n$ and $\ov\OO[m]_n$), let
$$N^{\rm sing}:=\{ v \in N \, | \,E_\beta v=0 \hbox{ for all\,\, $\beta \in \w{\Phi}^+$ (resp., $\Phi[m]^+_n$ and $\ov{\Phi}[m]^+_n$)}\}$$
stand for the subspace spanned by singular vectors in $N$. For any weight $\mu$ of $N$, the weight space $N_\mu$ is finite dimensional by definition of $\w\OO$  (resp., $\OO[m]_n$ and $\ov\OO[m]_n$). We denote $N_\mu^{\rm sing}:=N_\mu \cap N^{\rm sing}$.

By an argument similar to the proof of  \cite[Propositions 3.5 and 3.7]{CaL}, one can show that $\wcH^i$ (resp., $\cH^i[m]_n$ and $\ovcH^i[m]_n$) mutually commute with each other, and they are $\DG$(resp., $\G[m]_n$ and $\SG[m]_n$)-homomorphisms for $i=1, \ldots, \ell$. We immediately see that $M^{\rm sing}$ and the finite-dimensional subspace $M_\mu^{\rm sing}$ are $\wcH^i$-invariant for any weight $\mu$ of $M$. Thus, we may view $\wcH^i$ as a linear endomorphism on $M_\mu^{\rm sing}$. Similarly, $\cH^i[m]_n$ (resp., $\ovcH^i[m]_n$) may be viewed as a linear endomorphism on $(M_1\otimes\cdots\otimes  M_\ell)_\mu^{\rm sing}$ for $M_1, \ldots, M_\ell \in \OO[m]_n$ (resp., $\ov\OO[m]_n$) and any weight $\mu$ of $M_1\otimes\cdots\otimes  M_\ell$.

\begin{lem} \label{wOm-Om}
Let $N_1, N_2 \in \w{\mathcal {O}}$, and let $v\in N_1 \otimes N_2$ be a weight vector of weight $\mu$.
\begin{enumerate}[\normalfont(i)]
\item If $\mu\in \Xi[m]$, then $\w\Omega(v)=\Omega[m](v)$.
\item If $\mu\in \ov\Xi[m]$, then $\w\Omega(v)=\ov\Omega[m](v)$.
\end{enumerate}
\end{lem}
\begin{proof}
The proof is similar to that of  \cite[Lemma 3.11]{CaL} with a slight modification. For completeness, we include it here. We will only prove (i). The proof of (ii) is similar.
We may assume that $v=v_1 \otimes v_2$, where $v_i \in N_i$ is a weight vector of weight $\mu_i$ for $i=1, 2$, and $\mu_1+\mu_2=\mu$.
For $i=1, 2$, $\mu_i \in \Xi[m]$ by \lemref{weight-decomposition2}.
For all $k \in \hf \N \backslash \J_m^+$ and $i=1, 2$, $\mu_i(E_k)=0$, and so $E_k(v_i)=0$.
By virtue of \remref{O[m]}, it remains to consider $\beta \in \w{\Phi}^+$ with $\pm \beta \notin {\Phi}[m]^+$. For such $\beta$, we have $\beta(E_i)\not= 0$ for some $i\in \hf \N \backslash \J_m^+$. It follows that either the weight of $E_\beta v_1$ or $E^\beta v_2$ does not lie in $\wt{\Xi}$. Thus either $E_\beta v_1=0$ or $E^\beta v_2=0$, and hence $E_\beta\otimes E^\beta (v_1\otimes v_2)=0$. Similarly, $E^\beta\otimes E_\beta (v_1\otimes v_2)=0$. Therefore, $\w\Omega(v_1\otimes v_2)=\Omega[m] (v_1\otimes v_2)$.
\end{proof}

As a consequence, we obtain the following lemma.
\begin{lem} \label{corresp}
Let $M_1, \ldots, M_\ell \in \w\OO$, and let ${v} \in {M_1\otimes\cdots\otimes  M_\ell}$ be a weight vector of weight $\mu$.
\begin{enumerate}[\normalfont(i)]
\item If $\mu \in \Xi[m]$, then $\wcH^i {v}= \cH^i[m](v)$ for all $i=1, \ldots, \ell$.
\item If $\mu \in \ov{\Xi}[m]$, then $\wcH^i {v}=\ovcH^i[m](v)$ for all $i=1, \ldots, \ell$.
\end{enumerate}
\end{lem}

We would like to ask whether the eigenvectors of $\cH^i$, $\cH^i[m]$ and $\ovcH^i[m]$ are related. To answer the question, we need the following proposition. Recall the bijections $\w\la\leftrightarrow \la[m]\leftrightarrow\ov\la[m]$ in \eqref{cP}.

\begin{prop} \label{isom}
Let $\wt M \in \w\OO$, and let $\w\mu \in \w{\cP}^+$ be a weight of $\wt M$. Then:
\begin{enumerate}[\normalfont(i)]
\item There exists $A \in U (\tilde{\mf l})$ such that the map $\mathfrak{t}^{\w\mu}_{[m]}: \wt M_{\w\mu}^{\rm sing} \to T_{[m]}(\wt {M})_{\mu[m]}^{\rm sing}$, defined by
$\mathfrak{t}^{\w\mu}_{[m]}({v})=A {v}$ for $ v \in \wt M_{\w\mu}^{\rm sing}$, is a linear isomorphism.
\item  There exists $\bar{A} \in U (\tilde{\mf l})$ such that the map $\ov{\mathfrak{t}}^{\w\mu}_{[m]}: \wt M_{\w\mu}^{\rm sing} \to \ov{T}_{[m]}(\wt M)_{\ovmu[m]}^{\rm sing}$, defined by
$\ov{\mathfrak{t}}^{\w\mu}_{[m]}({v})=\bar{A} {v}$ for $ v \in \wt M_{\w\mu}^{\rm sing}$, is a linear isomorphism.
\end{enumerate}
\end{prop}

\begin{proof}
We will only prove (i). The proof of (ii) is similar.
Note that there is a linear isomorphism
 \begin{align*}
\Hom_{\w \OO} \left(\Delta(\DG, \w\mu),\wt{M} \right)  & \longrightarrow \wt M_{\w\mu}^{\rm sing}\\
\varphi &\ \mapsto \  \varphi(v_{\w\mu})
  \end{align*}
where $v_{\w\mu}$ is a highest weight vector of $\Delta(\DG, \w\mu)$.

On the other hand, there exists $A \in U(\tilde{\mf l})$ such that $v_{\mu[m]}:=A v_{\w\mu}$ is a highest weight vector in $\Delta(\G[m], \mu[m])=T_{[m]}(\Delta(\DG, \w\mu))$ with weight $\mu[m]$. In fact, $A$ is a product of elements in $\tilde{\mf l}$ corresponding to a sequence of odd reflections (see \cite[Section 3.1]{CL10} and \cite[Section 4]{CLW11} for details). Similarly, the map
 \begin{align*}
\Hom_{\OO[m]} \left(\Delta(\G[m], \mu[m]), T_{[m]}(\wt{M}) \right) &\longrightarrow T_{[m]}(\wt{M})_{\mu[m]}^{\rm sing}\\
\phi  &\ \mapsto \ \ \phi(v_{\mu[m]})
  \end{align*}
is a linear isomorphism. By \thmref{SD}, we have
$$\Hom_{\w \OO} \left(\Delta(\DG, \w\mu),\wt{M} \right) \cong \Hom_{\OO[m]} \left(\Delta(\G[m], \mu[m]), T_{[m]}(\wt {M}) \right),$$
and hence $\wt M_{\wt \mu}^{\rm sing} \cong  T_{[m]}(\wt{M})_{\mu[m]}^{\rm sing}$. We may also see that any vector ${v} \in \wt M_{\wt \mu}^{\rm sing}$ corresponds to $A {v} \in  T_{[m]}(\wt{M})_{\mu[m]}^{\rm sing}$ under the isomorphism, which shows that the isomorphism is indeed the map $\mathfrak{t}^{\w\mu}_{[m]}$ as stated.
\end{proof}

\begin{rem}\label{A-inverse}
\begin{enumerate}[\normalfont(i)]

\item The elements $A$ and $\bar{A}$ in \propref{isom} depend only on the weight $\w\mu$, but not on the module $M$.

\item There exist $B, \bar{B} \in U(\tilde{\mf l})$ such that the inverses of $\mathfrak{t}^{\w\mu}_{[m]}$ and $\ov{\mathfrak{t}}^{\w\mu}_{[m]}$ are given respectively by $(\mathfrak{t}^{\w\mu}_{[m]})^{-1}(v)=B v$ and $(\ov{\mathfrak{t}}^{\w\mu}_{[m]})^{-1}(w)=\bar{B} w$ for any $v \in  T_{[m]}(\wt{M})_{\mu[m]}^{\rm sing}$ and $w \in \ov{T}_{[m]}(\wt M)_{\ovmu[m]}^{\rm sing}$. Again $B$ and $\ov{B}$ are products of elements in $\tilde{\mf l}$ corresponding to  sequences of odd reflections.
\end{enumerate}
\end{rem}

\begin{thm} \label{eigenvector-C1}
For $\wt M_1, \cdots, \wt M_\ell \in \w\OO$, let $\wt M=\wt M_1\otimes\cdots\otimes \wt M_\ell$. Suppose that ${v} \in \wt M_{\w\mu}^{\rm sing}$ with $\w\mu \in \w{\cP}^+$. For any $m, m' \in \Z_+$, let $v_{m'}=\mathfrak{t}^{\w\mu}_{[m']}({v})$ and $\ov v_m=\ov{\mathfrak{t}}^{\w\mu}_{[m]}({v})$. For each $i=1, \ldots, \ell$, we have:
\begin{enumerate}[\normalfont(i)]
\item ${v}$ is an eigenvector of $\wcH^i$ with eigenvalue $c$ if and only if $v_{m'}$ is an eigenvector of $\cH^i[m']$ with eigenvalue $c$.

Moreover, $\wcH^i$ is diagonalizable on $\wt M^{\rm sing}_{\w\mu}$ if and only if $\cH^i[m']$ is diagonalizable on $T_{[m']}(\wt M)^{\rm sing}_{\mu[m']}$. In this case, they have the same spectrum.
\item $v$ is an eigenvector of $\wcH^i$ with eigenvalue $c$ if and only if $\ov v_m$ is an eigenvector of $\ovcH^i[m]$ with eigenvalue $c$.

 Moreover, $\wcH^i$ is diagonalizable on $\wt M^{\rm sing}_{\w\mu}$ if and only if $\ovcH^i[m]$ is diagonalizable on $\ov T_{[m]}(\wt M)^{\rm sing}_{\ov\mu[m]}$. In this case, they have the same spectrum.
 \end{enumerate}
 As a consequence, $\ovcH^i[m]$ is diagonalizable on $\ov T_{[m]}(\wt M)^{\rm sing}_{\ov\mu[m]}$  if and only if $\cH^i[m']$ is diagonalizable on $T_{[m']}(\wt M)^{\rm sing}_{\mu[m']}$. In this case, they have the same spectrum.
\end{thm}

\begin{proof} We will only prove (i). The proof of (ii) is similar.
We know that $v_{m'}=A {v}$ for some $A \in U (\tilde{\mf l})$  by \propref{isom}.
Suppose $\wcH^i {v} = c {v}$ for some $c \in \C$. By \lemref{corresp} together with the fact that $A {v}$ is a vector of weight $\mu[m'] \in \Xi[m']$, we have
$$
\cH^i[m'] (A {v})=\wcH^i (A {v}).
$$
As $\wcH^i$ is a $\w\G$-homomorphism, it follows that
$$
\cH^i[m'] (v_{m'})= A \wcH^i (v)= A (c {v})=c v_{m'}.
$$
Conversely, suppose $\cH^i[m'] (v_{m'}) = c v_{m'}$ for some $c \in \C$. Use the property of $\wcH^i$ being a $\w\G$-homomorphism and \lemref{corresp} again, we have
$$A \wcH^i {v}= \wcH^i (A {v})=\cH^i[m'] (A v)=c v_{m'}.$$
By \remref{A-inverse}, $A$ has an inverse, and we deduce that $\wcH^i ({v})= c {v}$. This proves the first part of (i). The second part is a direct consequence of the first part.
\end{proof}

\begin{prop} \label{trun_eigenvector}
\begin{enumerate}[\normalfont(i)]
\item For $M_1, \ldots, M_\ell \in \OO[m]$, let $M=M_1\otimes\cdots\otimes  M_\ell$. Suppose that ${v} \in M_{\mu}^{\rm sing}$ with $\mu \in \cP[m]^+_n$.  Then for each $i=1, \ldots,\ell$, ${v}$ is an eigenvector of $\cH^i[m]$ with eigenvalue $c$ if and only if $v$ is an eigenvector of $\cH^i[m]_n$ with eigenvalue $c$.

Moreover, $\cH^i[m]$ is diagonalizable on $M_{\mu}^{\rm sing}$ if and only if $\cH^i[m]_n$ is diagonalizable on $\mf{tr}^\infty_n(M)^{\rm sing}_{\mu}$. In this case, they have the same spectrum.
\item For $\ov M_1, \ldots, \ov M_\ell \in \ov\OO[m]$, let $\ov M=\ov M_1\otimes\cdots\otimes  \ov M_\ell$. Suppose that ${v} \in \ov M_{\ov\mu}^{\rm sing}$ with $\ov\mu \in \ov\cP[m]^+_n$.  For each $i=1, \ldots, \ell$, ${v}$ is an eigenvector of $\ovcH^i[m]$ with eigenvalue $c$ if and only if $v$ is an eigenvector of $\ovcH^i[m]_n$ with eigenvalue $c$.

Moreover, $\ovcH^i[m]$ is diagonalizable on $\ov M_{\mu}^{\rm sing}$ if and only if $\ovcH^i[m]_n$ is diagonalizable on $\ov{\mf{tr}}^\infty_n(\ov M)^{\rm sing}_{\mu}$. In this case, they have the same spectrum.
\end{enumerate}
\end{prop}

\begin{proof}
We will only prove (i). The proof of (ii) is similar. Note that $\mf{tr}^\infty_n(M)^{\rm sing}_{\mu}=M_{\mu}^{\rm sing}$ for $\mu \in \cP[m]^+_n$. By \lemref{finitesum}, we have
 $$
\cH^i[m](w)=\cH^i[m]_n(w),\qquad \hbox{ for all $w\in M_{\mu}^{\rm sing}$}.
 $$
The first part of (i) follows. The second part of (i) follows readily from the first part.
\end{proof}

\section{Gaudin Hamiltonians on modules over $\ovcG[m]_n$} \label{diag}
In this section, we consider the (quadratic) Gaudin Hamiltonians for finite-dimensional Lie (super)algebras. We relate the Gaudin Hamiltonians for $\cG[m]_n$ (resp., $\ovcG[m]_n$) to those for $\G[m]_n$ (resp., $\SG[m]_n$) for $n \in \N$. Furthermore, we study the Gaudin Hamiltonians on the tensor product of unitarizable irreducible highest weight modules and give an affirmative answer to the diagonalization of these operators.

Let us fix $\ell \in \N$ with $\ell \geq 2$. For $n \in \N$, the Casimir symmetric tensors for $\cG[m]_n$ and $\ovcG[m]_n$ are defined by (cf. \cite[Section 3.4]{CaL})
 \begin{align*}
\mathring{\Omega}[m]_n&=\sum_{\beta\in {\Phi[m]}^+_n} \big(E^\beta\otimes E_{\beta}+(-1)^{|E_\beta|} E_\beta\otimes E^{\beta}) +\sum_{j\in{\J}^+_m(n)} (-1)^{2j} E_j\otimes E_j,\\
\mathring{\ov{\Omega}}[m]_n&=\sum_{\beta\in \ov{\Phi}[m]^+_n}(E^\beta\otimes E_{\beta}+(-1)^{|E_\beta|} E_\beta \otimes E^{\beta})  +\sum_{j\in\ov{\J}^+_m(n)}(-1)^{2j}E_j\otimes E_j.
  \end{align*}
Clearly, $\mathring{{\Omega}}[m]_n$ and $\mathring{\ov{\Omega}}[m]_n$ lie in $U(\cG[m]_n)\otimes U(\cG[m]_n)$ and $U(\ovcG[m]_n)\otimes U(\ovcG[m]_n)$, respectively.

For any $i=1, \ldots, \ell$ and any  distinct complex numbers $z_1, \ldots, z_n$, the (quadratic) Gaudin Hamiltonians $H^i[m]_n$ and $\xov H^i[m]_n$ are defined by
\begin{equation}\label{Hamiltonian}
H^i[m]_n =\sum_{\substack{j=1 \\ j\not=i}}^\ell \frac{{\mathring{\Omega}[m]_n}^{\! \! \! \!  (ij)}}{z_i-z_j}\qquad\hbox{and} \qquad
\xov H^i[m]_n =\sum_{\substack{j=1 \\ j\not=i}}^\ell \frac{{\mathring{\ov{\Omega}}[m]_n}^{\! \! \! \! \!    (ij)}}{z_i-z_j}.
\end{equation}

Let $M=M_1\otimes \cdots \otimes M_\ell$ (resp., $\ov M=\ov M_1\otimes \cdots  \otimes \ov M_\ell$), where each $M_i$ (resp., $\ov M_i$) is a $\cG[m]_n$(resp., $\ovcG[m]_n$)-module.
Note that the Gaudin Hamiltonians $H^i[m]_n$ on $M$ (resp., $\xov{H}^i[m]_n$ on $\ov M$) mutually commute with each other, and they are $\cG[m]_n$(resp., $\ovcG[m]_n$)-homomorphisms.
It is also evident that for any weight $\mu$ of $M$ (resp., $\ov M$), the subspace $M_{\mu}^{\rm sing}$ (resp., ${\ov M}_{\mu}^{\rm sing}$) is ${H}^i[m]_n$-invariant (resp., $\xov{H}^i[m]_n$-invariant). 
Here and below, $N^{\rm sing}$ stands for the subspace spanned by singular vectors in $N$ with respect to the standard Borel subalgebra
and $N_\mu^{\rm sing}$ for the subspace spanned by singular vectors in the weight space $N_\mu$ for any $\cG[m]_n$(resp., $\ovcG[m]_n$)-module $N$ and any weight $\mu$ of $N$.

Let $\xx=\mf{a, c, d}$. Each $\G^\xx[m]_n$(resp., $\ov\G^\xx[m]_n$)-module can be regarded as a $\cG^\xx[m]_n$(resp., $\ovcG^\xx[m]_n$)-module via the isomorphism $\iota$ defined by \eqref{iso-e}. It is clear that the set of singular vectors in a module $N$ regarded as a $\G^\xx[m]_n$(resp., $\ov\G^\xx[m]_n$)-module equals the set of singular vectors in the module $N$ regarded as a $\cG^\xx[m]_n$(resp., $\ovcG^\xx[m]_n$)-module.

\begin{prop}\label{iota_corresp}
\begin{enumerate}[\normalfont(i)]
\item For $n\in \N$ and $i=1, \ldots, \ell$, let $M_i$ be a $\G^\xx[m]_n$-module, and let $M=M_1\otimes \cdots \otimes M_\ell$ and $v \in M^{\rm sing}$. Then for each $i=1, \ldots, \ell$, $v$ is an eigenvector of $\cH^i[m]_n$ if and only if $v$ is an eigenvector of $H^i[m]_n$. Thus, $\cH^i[m]_n$ is diagonalizable on $M^{\rm sing}$ if and only if $H^i[m]_n$ is diagonalizable on $M^{\rm sing}$.

\item For $n\in \N$ and $i=1, \ldots, \ell$, let $\ov M_i$ be a $\ov\G^\xx[m]_n$-module, and let $\ov M=\ov M_1\otimes \cdots  \otimes \ov M_\ell$ and $v \in \ov M^{\rm sing}$. Then for each $i=1, \ldots, \ell$, $v$ is an eigenvector of $\ovcH^i[m]_n$  if and only if $v$ is an eigenvector of $\xov H^i[m]_n$. Thus, $\ovcH^i[m]_n$ is diagonalizable on $\ov M^{\rm sing}$ if and only if $\xov H^i[m]_n$ is diagonalizable on $\ov M^{\rm sing}$.
\end{enumerate}
\end{prop}

\begin{proof}
It is easy to see that
$\iota\otimes \iota \big(\mathring{\Omega}[m]_n\big)=\Omega[m]_n-(m-n)K\otimes K$. Therefore
$$
H^i[m]_n(v)=\cH^i[m]_n(v)-(m-n)\sum_{\substack{j=1 \\ j\not=i}}^\ell \frac{K^{(i)}  K^{(j)}}{z_i-z_j}(v).
$$
The last term on the right hand side is a fixed scalar times $v$. This implies (i). The proof of (ii) is similar.
\end{proof}

Let $n\in \N \cup \{\infty\}$. For $\xi_1, \ldots, \xi_\ell \in {\bf h}^\xx[m]_n^*$ (resp., $\ov\xi_1, \ldots, \ov\xi_\ell \in \ov{\bf h}^\xx[m]_n^*$), we set $\un{\xi}:=(\xi_1, \ldots, \xi_\ell)$ (resp., $\un{\ov{\xi}}:=(\ov\xi_1, \ldots, \ov\xi_\ell)$) and
$$
L(\cG^\xx[m]_n, \un{\xi}):=L(\cG^\xx[m]_n, \xi_1) \otimes \cdots \otimes L(\cG^\xx[m]_n, \xi_\ell)
$$
$$
\big(\textrm{resp., }L(\ovcG^\xx[m]_n, \un{\ov{\xi}}):=L(\ovcG^\xx[m]_n, \ov\xi_1) \otimes \cdots \otimes L(\ovcG^\xx[m]_n, \ov\xi_\ell) \big).
$$
Similarly, for $\xi_1, \ldots, \xi_\ell \in \fh^\xx[m]_n^*$ (resp., $\ov\xi_1, \ldots, \ov\xi_\ell \in \ov{\fh}^\xx[m]_n^*$), we set $\un{\xi}:=(\xi_1, \ldots, \xi_\ell)$ (resp., $\un{\ov{\xi}}:=(\ov\xi_1, \ldots, \ov\xi_\ell)$) and
$$
L(\G^\xx[m]_n, \un{\xi}):=L(\G^\xx[m]_n, \xi_1) \otimes \cdots \otimes L(\G^\xx[m]_n, \xi_\ell)
$$
$$
\big(\textrm{resp., }L(\SG^\xx[m]_n, \un{\ov{\xi}}):=L(\SG^\xx[m]_n, \ov\xi_1) \otimes \cdots \otimes L(\SG^\xx[m]_n, \ov\xi_\ell) \big).
$$
Again we will drop the subscript $\infty$. For instance, we denote $L(\G^\xx[m], \un{\xi}):=L(\G^\xx[m]_\infty, \un{\xi})$.

The following theorem follows from \cite[Main Corollary]{Ry}.

\begin{thm} \label{R}
Let $n\in\N$, and let $\xi_1, \ldots, \xi_\ell \in {\bf h}^\xx[0]_n^*$ be dominant integral weights. For each $i=1, \ldots, \ell$ and generic $z_1, \ldots, z_\ell$, the Gaudin Hamiltonian $H^i[0]_n$ is diagonalizable on the space $L(\cG^{\xx}[0]_n, \un{\xi})^{\rm sing}$.
\end{thm}

\begin{rem}
The results in \cite{Ry} involve only simple Lie algebras. The Lie algebra $\cG^{\mf{a}}[0]_n \cong \mathfrak{gl}(n)$ is, however, not semisimple. The case where $\xx=\mf{a}$ is still true since the irreducible highest weight modules over $\mathfrak{gl}(n)$ coincide with those over $\mathfrak{sl}(n)$, and each Gaudin Hamiltonian for $\mathfrak{gl}(n)$ minus that for $\mathfrak{sl}(n)$ is a scalar multiple of the identity operator on the space $L(\cG^{\mf{a}}[0]_n, \un{\xi})^{\rm sing}$.
\end{rem}

Recall $Q^\xx(m|n)$, $\ov Q^\xx(m|n)$, $\ov{\mc Q}^{\xx, I}(m|n)$ and $\ov{\mc Q}^{\xx, II}(m|n)$ defined in Sections \ref{unitarizable-a} and \ref{unitarizable-cd}.

\begin{cor}\label{cor:R}
Let $\xi_1, \ldots, \xi_\ell \in Q^\xx(0|\infty)$. For each $i=1, \ldots, \ell$ and generic $z_1, \ldots, z_\ell$, the Gaudin Hamiltonian $\cH^i[0]$ is diagonalizable on the space $L(\G^{\xx}[0],\un{\xi})^{\rm sing}$.
\end{cor}

\begin{proof} Note that $\mf{tr}^\infty_n(L(\G^{\xx}[0],{\xi}_i))$, regarded as a $\cG^\xx[0]_n$-module via the isomorphism $\iota$ defined by \eqref{iso-e}, is an irreducible module with dominant integral highest weight. Given a weight space of $L(\G^{\xx}[0], \un{\xi})$, we can choose $n$ large enough that $\mf{tr}^\infty_n(L(\G^{\xx}[0], \un{\xi}))$ contains the given weight space. Thus the corollary follows from \propref{trun_eigenvector}, \propref{iota_corresp} and \thmref{R}.
\end{proof}

Fix $m\in \Z_+$. According to the bijections \eqnref{Qa} and \eqnref{Qx}, any weight $\xi \in  Q^\xx(0|\infty)$ corresponds to a weight $\ov\xi\in \ov Q^\xx(m|\infty)$, and vice versa.

\begin{thm}\label{thm:R-D}
Let $\xi_1, \ldots, \xi_\ell \in  Q^\xx(0|\infty)$, and let $\ov\xi_1, \ldots, \ov\xi_\ell  \in \ov Q^\xx(m|\infty)$ be the corresponding weights. For each $i=1, \ldots, \ell$ and generic $z_1, \ldots, z_\ell$, the Gaudin Hamiltonian $\ovcH^i[m]$ is diagonalizable on the space $L(\SG^{\xx}[m], \un{\ov{\xi}})^{\rm sing}$, and any eigenbasis of $\ovcH^i[m]$ on $L(\SG^{\xx}[m], \un{\ov{\xi}})^{\rm sing}$ can be obtained from some eigenbasis of $\cH^i[0]$ on $L(\G^{\xx}[0], \un{\xi})^{\rm sing}$, and vice versa. Moreover, the actions of $\ovcH^i[m]$ on $L(\SG^{\xx}[m], \un{\ov{\xi}})^{\rm sing}$ and $\cH^i[0]$ on $L(\G^{\xx}[0], \un{\xi})^{\rm sing}$ have the same spectrum.
\end{thm}

\begin{proof}
As noted in Section \ref{category-O}, $T_{[0]}$ and $\ov T_{[m]}$ are tensor functors. For $n=\infty$, the theorem follows from \corref{cor:R} and \thmref{eigenvector-C1} with $m'=0$ together with \eqnref{Qa} and \eqnref{Qx}.
\end{proof}

\begin{thm} \label{dom-uni-1}
Let $n\in \N$ and $\xi_1, \ldots, \xi_\ell \in \ov{\mc Q}^{\xx, I}(m|n)$. For each $i=1, \ldots, \ell$ and generic $z_1, \ldots, z_\ell$, the Gaudin Hamiltonian $\xov H^i[m]_n$ is diagonalizable on the space $L(\ovcG^{\xx}[m]_n, \un{\xi})^{\rm sing}$, and any of its eigenbasis can be obtained from the singular eigenvectors of the Gaudin Hamiltonians $H^i[0]_k$'s on some tensor products of finite-dimensional irreducible $\cG^\xx[0]_k$-modules.
\end{thm}

\begin{proof}
We prove the theorem for $\xx=\mf{c}$. The other cases can be proved similarly.

For any $i=1, \ldots, \ell$, $\xi_i=\ov{\la}_i-d_i \mathbbm{1}_{m|n}$ for some $\la_i \in \cP(m|n)$ and $d_i \in \Z_+$ such that $(\la_{i})_1 \leq d_i$. Regarded as a $\SG^{\mf{c}}[m]_n$-module via the isomorphism $\iota$ given in \eqnref{iso-e}, $L(\ovcG^{\mf{c}}[m]_n, \un{\xi})$ is isomorphic to $L(\SG^{\mf{c}}[m]_n, \ovla_1[m])\otimes\cdots \otimes L(\SG^{\mf{c}}[m]_n, \ovla_\ell[m])$, where
$$
\ovla_i[m] :=\sum_{j=1}^{m} (\la_{i})_j\ep_{j}+ \sum_{j=1}^n \left\langle \left(\la_i\right)^\prime_j -m \right\rangle \ep_{j-\hf}+ d_i \La_0  \qquad \quad \hbox{for $i=1, \ldots, \ell$.}
$$
Therefore $\ovla_1[m], \ldots, \ovla_\ell[m]\in \ov Q^{\mf{c}}(m|\infty)$. Now the theorem follows from \propref{trun_eigenvector}, \propref{iota_corresp} and \thmref{thm:R-D}.
\end{proof}

 \begin{rem}
 \begin{enumerate}[\normalfont(i)]
 
 \item The module $L(\ovcG^\xx[m]_n, \un{\xi})$ is a direct sum of irreducible highest weight modules since it is a tensor product of unitarizable irreducible highest weight modules. As each $\xov H^i[m]_n$ commutes with $\ovcG^{\xx}[m]_n$, we see that the Gaudin Hamiltonian $\xov H^i[m]_n$ is diagonalizable on the space $L(\ovcG^{\xx}[m]_n, \un{\xi})$.

 \item  For $\xx=\mf{a}$ and each $L(\ovcG^{\mf{a}}[m]_n, {\xi}_i)$ being the natural module $\C^{m|n}$, the corresponding result in \thmref{dom-uni-1} has been obtained by Mukhin, Vicedo and Young  \cite{MVY}.

\item We have $\ovcG^{\mf{c}}[0]_n \cong \mf{so}(2n)$ and $\ovcG^{\mf{d}}[0]_n \cong \mf{sp}(2n)$.  The weights $\xi_1, \ldots, \xi_\ell$ in \thmref{dom-uni-1} are highest weights of infinite-dimensional unitarizable irreducible highest weight modules (see \remref{inf_U}).

\item \thmref{dom-uni-1} is also valid for the ortho-symplectic Lie superalgebra $\osp(2m+1|2n)$ if $\xi_1, \ldots, \xi_\ell$ are the highest weights such that for each $i=1, \ldots, \ell$, $\xi_i=\ov{\la}_i-d_i \mathbbm{1}_{m|n}$ for some $\la_i \in \cP(m|n)$ and $d_i \in \Z_+$ satisfying $(\la_{i})_1 \leq d_i$. The proof is identical and is omitted here.

\end{enumerate}
\end{rem}

\begin{thm}  \label{dom-uni-2}
Let $\xx=\mf{c, d}$, $n\in \N$ and $\xi_1, \ldots, \xi_\ell \in \ov{\mc Q}^{\xx, II}(m|n)$. For each $i=1, \ldots, \ell$ and generic $z_1, \ldots, z_\ell$, the Gaudin Hamiltonian $\xov H^i[m]_n$ is diagonalizable on the space $L(\ovcG^{\xx}[m]_n, \un{\xi})^{\rm sing}$, and any of its eigenbasis can be obtained from the singular eigenvectors of the Gaudin Hamiltonians $H^i[0]_k$'s on some tensor products of finite-dimensional irreducible $\cG^\xx[0]_k$-modules.

\end{thm}

\begin{proof} Each $\ovcG^{\mf c}[m]_n$(resp., $\ovcG^{\mf d}[m]_n$)-module $L(\ovcG^{\mf c}[m]_n,{\xi}_i)$ (resp., $L(\ovcG^{\mf d}[m]_n,{\xi}_i)$) can be regarded as a $\cG^{\mf d}[m]_n$(resp., $\cG^{\mf c}[m]_n$)-module via the isomorphism $\widehat{\varphi}$ defined by \eqref{phi_cd}. The theorem follows from \thmref{dom-uni-1} together with an explicit description of the highest weights involved.
\end{proof}

We anticipate that the results of this paper may provide an approach of obtaining the common eigenvectors of the Gaudin Hamiltonians associated to the Lie (super)algebra $\ovcG^\xx[m]_n$ from the common eigenvectors of the Gaudin Hamiltonians associated to the Lie algebra $\cG^\xx[0]_k$ for some $k \in \N$.

Let us explain in more detail for $\xx=\mf{a}$.
For each $i=1, \ldots, \ell$, consider the Gaudin Hamiltonian $\xov H^i[m]_n$ on the space $\ov M_{\ovmu}^{\rm sing}$, where $\ovmu$ is a singular weight ({\rm i.e.} a weight of a singular vector) of $\ov M$ and $\ov M:=L(\ovcG^{\mf{a}}[m]_n, \ov\xi_1) \otimes \cdots \otimes L(\ovcG^{\mf{a}}[m]_n, \ov\xi_\ell)$ for $\ov\xi_1, \ldots, \ov\xi_\ell \in \ov{\mc Q}^{\mf{a}, I}(m|n)$. The arguments in Sections \ref{GH} and \ref{diag} allow us to obtain the common eigenvectors of all $\xov H^i[m]_n$'s from the common eigenvectors of all Gaudin Hamiltonians ${H}^i[0]_k$'s for the Lie algebra $\cG^\mf{a}[0]_k$, for some $k \in \N$, on the space spanned by singular vectors in $L_1 \otimes \cdots \otimes L_\ell$, where $L_i$ is a finite-dimensional irreducible $\cG^\mf{a}[0]_k$-module depending on $L(\ovcG^{\mf{a}}[m]_n, \ov\xi_i)$ for $i=1, \ldots, \ell$.
An explicit construction would be to apply a sequence of certain odd reflections to the Bethe vectors of ${H}^i[0]_k$, which are constructed by Bethe ansatz method (see, for example, \cite[p. 31]{FFR}). These odd reflections depend on the weight $\mu$ and can be determined explicitly.
Let us say we would like to find common eigenvectors for the Gaudin Hamiltonians $\xov H^i[1]_n$'s on $\ov M_{\ovmu}^{\rm sing}$ for the singular weight $\ovmu:=3 \ep_1+3\ep_{1/2}$.
For type $\mf a$, we have $d=0$, and hence we may identify the action of $\ovcG^{\mf{a}}[1]_n$ (resp., $\xov H^i[1]_n$) with the action of $\SG^{\mf{a}}[1]_n$ (resp., $\ov{\mc H}^i[1]_n$) and $\ov{\mc Q}^{\mf{a}, I}(1|n)$ with $\ov{ Q}^{\mf a, I}(1|n)$. Therefore $\ov M=L(\SG^{\mf{a}}[1]_n, \ov\xi_1) \otimes \cdots \otimes L(\SG^{\mf{a}}[1]_n, \ov\xi_\ell)$.
Denote by $\mu$ (resp., $\xi_i$) the partition associated to the weight $\ovmu$ (resp., $\ov\xi_i$).
The common eigenvectors of the $\ov{\mc H}^i[1]_n$'s on $\ov M_{\ovmu}^{\rm sing}$ correspond to the common eigenvectors of the $\wcH^i$'s on $(L(\w\G^{\mf{a}}, \w\xi_1) \otimes \cdots \otimes L(\w\G^{\mf{a}}, \w\xi_\ell))^{\rm sing}_{\w\mu}$ and thus correspond to the common eigenvectors of $\cH^i[0]$'s on $(L(\G^{\mf{a}}[0], \xi_1[0]) \otimes \cdots \otimes L(\G^{\mf{a}}[0], \xi_\ell[0])^{\rm sing}_{\mu[0]}$, where  
$\mu[0]=3\ep_1+\ep_2+\ep_3+\ep_4\leftrightarrow\w\mu\leftrightarrow\ov\mu=\ovmu[1]$ and
$\xi_i[0]\leftrightarrow\w\xi_i\leftrightarrow\ov\xi_i=\ov\xi_i[1]$ 
are determined by the bijections in \eqref{Qa}. 
More precisely, we may choose the common eigenvectors $\ov v:=E_{1/2,2} E_{1/2,3} E_{1/2,4} v$ for all Bethe vectors $v$ of weight $\mu[0]$ in 
$(L(\G^{\mf{a}}[0]_k, \xi_1[0])\otimes\cdots\otimes L(\G^{\mf{a}}[0]_k, \xi_\ell[0]))^{\rm sing}=(L(\cG^{\mf{a}}[0]_k, \xi_1[0])\otimes\cdots\otimes L(\cG^{\mf{a}}[0]_k, \xi_\ell[0]))^{\rm sing}$ for some $k$.
We emphasize that although all of these $E_{i,j}$'s, coming from odd reflections, lie in $U(\w\G^\mf{a})$, each resulting vector $\ov v$ is an element of ${\ov M}_{\ovmu}^{\rm sing}$. In this case, we can choose $k=4$.

One should expect that the eigenvectors obtained in this way take complicated forms in general. It would be interesting to know whether this procedure is related to any other known method of constructing eigenvectors.

\vskip 0.5cm
\noindent{\bf Acknowledgments.}
We would like to thank an anonymous expert for helpful comments.
The first author was partially supported by Ministry of Science and Technology of Taiwan grant 110-2115-M-006-006.
The second author was partially supported by Ministry of Science and Technology of Taiwan grant 109-2115-M-006-019-MY3.

\bigskip
\noindent
Department of Mathematics, National Cheng Kung University, Tainan 701401, Taiwan\\
{\it E-mail address}: \texttt{keng@ncku.edu.tw}

\bigskip
\noindent
Department of Mathematics, National Cheng Kung University, Tainan 701401, Taiwan\\
{\it E-mail address}: \texttt{nlam@ncku.edu.tw}

\end{document}